\documentclass[conference]{IEEEtran}
\IEEEoverridecommandlockouts
\usepackage{cite}
\usepackage{amsmath,amssymb,amsfonts,amsthm}
\usepackage{booktabs,multirow,nicefrac}
\usepackage{graphicx}
\usepackage{textcomp}
\usepackage[hyphens]{url}
\usepackage[normalem]{ulem}
\usepackage{fancyhdr}
\usepackage{pifont}
\IfFileExists{bbding.sty}{%
  \usepackage{bbding}%
}{%
  \providecommand{\Checkmark}{\ding{51}}%
  \providecommand{\XSolidBrush}{\ding{55}}%
}
\usepackage{comment}
\usepackage{xspace}
\usepackage{etoolbox}
\usepackage{enumitem,kantlipsum}
\usepackage{makecell}

\usepackage[dvipsnames,svgnames,table]{xcolor}
\usepackage{soul}
\sethlcolor{yellow} 

\usepackage{algorithm}
\usepackage{algorithmic}
\usepackage{tikz}
\usetikzlibrary{calc}
\usepackage[most]{tcolorbox}
\usepackage[skipabove=0pt,skipbelow=0pt,leftmargin=0pt,rightmargin=0pt]{mdframed}
\usepackage[colorinlistoftodos]{todonotes}

\usepackage{subcaption}
\usepackage{hyperref}

\renewcommand{\baselinestretch}{0.92}

\def\BibTeX{{\rm B\kern-.05em{\sc i\kern-.025em b}\kern-.08em
    T\kern-.1667em\lower.7ex\hbox{E}\kern-.125emX}}

\newtheorem{theorem}{Theorem}

\DeclareRobustCommand{\ourmodel}{\textit{FEnc}$^2$\xspace}
\DeclareRobustCommand{\multi}{\textit{Arch-aware Ct Compression}\xspace}
\DeclareRobustCommand{\frag}{\textit{Conv-aware Encoding}\xspace}

\soulregister{\cite}{1}
\soulregister{\ref}{1}
\soulregister{\textit}{1}
\soulregister{\textbf}{1}
\soulregister{\frag}{0}
\soulregister{\multi}{0}
\soulregister{\ourmodel}{0}
\soulregister{\xspace}{0}
\soulregister{\noindent}{0}

\newmdenv[
  backgroundcolor=blue!10,
  linecolor=blue!50,
  linewidth=0.5pt,
  innerleftmargin=5pt,
  innerrightmargin=5pt,
  innertopmargin=5pt,
  innerbottommargin=5pt,
  breakable=true
]{bluebox}

\newmdenv[
  backgroundcolor=green!10,
  linecolor=green!50,
  linewidth=0.5pt,
  innerleftmargin=5pt,
  innerrightmargin=5pt,
  innertopmargin=5pt,
  innerbottommargin=5pt,
  breakable=true
]{greenbox}

\begin{document}
\bstctlcite{IEEEtran:BSTcontrol}

% 页面设置
\pdfpagewidth=8.5in
\pdfpageheight=11in
\pagenumbering{arabic}

\newcommand{\iscasubmissionnumber}{147}
\title{\LARGE{\ourmodel: Unifying Data Packing for Efficient Private Inference via Convolution and Architecture-Aware Fragment Encoding}}
%\vspace{-0.2in}
% \author{\normalsize{ISCA 2026 Submission \textbf{\#\iscasubmissionnumber} -- Confidential Draft -- Do NOT Distribute!!}}

% \author{Zhaoting Gong}
% \email{zgong6@ncsu.edu}
% \affiliation{%
%   \institution{North Carolina State University}
%   \city{Raleigh, NC}

% \author{Ran Ran}
% \email{rran@ncsu.edu}
% \affiliation{%
%   \institution{North Carolina State University}
%   \city{Raleigh, NC}

% \author{Fan Yao}
% \email{fan.yao@ucf.edu}
% \affiliation{%
%   \institution{University of Central Florida}
%   \city{Orlando, FL}

% \author{Wujie Wen}
% \email{wwen2@ncsu.edu}
% \affiliation{%
%   \institution{North Carolina State University}
%   \city{Raleigh, NC}

\author{
\IEEEauthorblockN{
Ran Ran\textsuperscript{1},
Zhaoting Gong\textsuperscript{1},
Nuo Xu\textsuperscript{2},
Yuanchao Xu\textsuperscript{3},
Fan Yao\textsuperscript{4},
and Wujie Wen\textsuperscript{1}
}
\IEEEauthorblockA{
\textsuperscript{1}North Carolina State University \\
\textsuperscript{2}University of Minnesota \\
\textsuperscript{3}University of California, Santa Cruz \\
\textsuperscript{4}University of Central Florida \\
\textsuperscript{1}\{rran, zgong6, wwen2\}@ncsu.edu \\
\textsuperscript{2}xu001536@umn.edu,
\textsuperscript{3}yxu314@ucsc.edu,
\textsuperscript{4}fan.yao@ucf.edu
}
}

\maketitle
\thispagestyle{plain}
\pagestyle{plain}

%%%%%% -- PAPER CONTENT STARTS-- %%%%%%%%

\begin{abstract}
% Homomorphic Encryption (HE) enables privacy-preserving machine learning but suffers from extreme computational and memory overhead. These costs arise not only from slow low-level primitives such as NTTs, rotations, and key-switching, but also from suboptimal ciphertext packing at the application level. Existing layouts preserve either spatial locality or channel locality, but not both, resulting in wasted ciphertext slots, excessive rotations, and inflated ciphertext counts.
Fully Homomorphic Encryption (FHE) enables privacy-preserving machine learning but incurs extreme computational and memory overhead. These costs stem not only from slow low-level primitives such as Number Theoretic Transform (NTT), rotation, and key-switching, but also from inefficient ciphertext packing at the application level. Existing packing strategies typically preserve either neighboring data elements or feature-grouping information, but not both, leading to wasted ciphertext slots, excessive rotations, and inflated ciphertext counts. We propose \ourmodel, a unified and principled fragment-based encoding framework that optimizes slot utilization, rotation complexity, and ciphertext density for CKKS-based private convolutional neural network inference. Rather than applying static or layer-isolated heuristics, \ourmodel introduces (1) \frag, which analytically selects an optimal fragment (block) size to decouple spatial dependencies and jointly minimize inner–outer rotations across layers, and (2) \multi, which dynamically restores ciphertext density after feature- or channel-reduction layers. Together, these transformations reshape encrypted workload structure, reducing homomorphic operations by one to two orders of magnitude. With full memory capacity utilized (i.e., at maximum batch size), \ourmodel achieves end-to-end latency speedups over the state-of-the-art \texttt{Orion} of up to \textbf{228.83$\times$} (GPU) and \textbf{226.06$\times$} (CPU) for LeNet (MNIST), and up to \textbf{4.55$\times$} (GPU) and \textbf{9.43$\times$} (CPU) for MobileNet (ImageNet).
%\ourmodel achieves up to \textbf{$228.83\times$} latency speedup over \texttt{Orion} on GPU and up to \textbf{$226.06\times$} on CPU for LeNet (MNIST). For MobileNet (ImageNet), \ourmodel achieves up to \textbf{$4.55\times$} latency speedup over \texttt{Orion} on GPU and up to \textbf{$9.43\times$} on CPU. 
Importantly, \ourmodel is hardware-agnostic but architecturally transformative: by optimizing encrypted tensor layout before execution, it reduces ciphertext count and workload pressure on hardware, complementing primitive-level optimizations (e.g., NTT/keyswitch accelerators). This demonstrates that application-level data layout is a first-order architectural design dimension for encrypted inference and a critical enabler for next-generation FHE systems.

\end{abstract}


\section{Introduction}

Deep Neural Networks (DNNs) underpin a wide range of modern computer vision tasks, including image classification and object detection~\cite{krizhevsky2010convolutional}. Machine Learning as a Service (MLaaS) platforms (e.g., Amazon SageMaker~\cite{amazonsagemaker}, Google AI Platform~\cite{bisong2019overview}, Azure ML~\cite{team2016azureml}, and OpenAI~\cite{openaiplatform}) enable scalable deployment of such models, but raise significant privacy concerns when sensitive data is processed in untrusted cloud.
Fully Homomorphic Encryption (FHE)~\cite{gentry2009fully, brakerski2014leveled, cheon2017homomorphic, chillotti2020tfhe} enables computation directly over encrypted data, providing a foundation for privacy-preserving machine learning. 

Despite its promise, encrypted convolutional neural network (CNN) inference remains orders of magnitude slower than its plaintext counterpart, even on modern GPUs and HE accelerators. For instance, Orion, the state-of-the-art HE inference system~\cite{ebel2025orion}, takes more than 300 seconds to infer one single encrypted CIFAR-10 image on an Intel Xeon-based server, which represents orders of magnitude slower than plaintext inference.
The bottleneck arises not only from the inherent latency of HE primitives, e.g., rotation, key-switching, and number-theoretic transforms (NTT), but also from the \emph{volume and structure} of these operations, which depend on how activations are packed into ciphertexts across layers (i.e., ciphertext packing). By packing scalar values into the vector slots of a single ciphertext, it enables Single Instruction Multiple Data (SIMD)-style HE operations such as SIMD additions and SIMD multiplications, amortizing the high cost of HE primitives across multiple data elements~\cite{gilad2016cryptonets,brutzkus2019low,dathathri2019chet,kim2022secure,lee2022low}.
%Packing determines SIMD slot utilization and the frequency of data realignment via rotations, exerting an outsized influence on end-to-end performance.

%all intermediate values remain concealed from the cloud. 

%Homomorphic Encryption (HE) addresses this by allowing computations to be performed directly on encrypted data~\cite{gentry2009fully, brakerski2014leveled, cheon2017homomorphic, chillotti2020tfhe}, ensuring that both inputs and intermediate results remain confidential, with only the client possessing the decryption key.

%\textcolor{red}{original} While enticing, the practicality of HE-based private CNN inference remains limited by extreme latency overheads~\cite{gilad2016cryptonets, rivest1978data, jung2021accelerating}. Orion~\cite{ebel2025orion}, the state-of-the-art HE inference system, takes 337.2s to classify a single encrypted CIFAR-10 image using AlexNet on an Intel Xeon Platinum 8581C CPU--several orders of magnitude slower than plaintext inference. The primary bottleneck is HE convolution, which dominates end-to-end latency, accounting for 95\% in AlexNet (no bootstrapping) and 65\% in ResNet-20 (with bootstrapping)\cite{ebel2025orion}, due to ciphertext expansion (by hundreds to thousands of times) and the high cost of HE operations (e.g., ciphertext rotation, multiplication, and addition), which are orders of magnitude slower than their plaintext counterparts\cite{cheon2017homomorphic}.

We argue that the performance of ciphertext packing critically depends on two factors: \textit{(1) the number of ciphertext rotations}, which cyclically shift encrypted vector elements for SIMD-style aggregation, and \textit{(2) slot utilization}, which measures how efficiently ciphertext slots are filled with useful data, reflecting hardware-level SIMD efficiency. HE-CNN inference exacerbates both factors: (i) convolutions introduce nested data dependencies, requiring extensive inner- and outer-rotations for intra- and inter-channel aggregation (Sec.~\ref{sec:data_den}). Each rotation incurs costly key-switching and multiple NTTs, similar to large-scale vector shuffles~\cite{samardzic2022craterlake}, and can account for over 70\% of end-to-end latency at the application level (Fig.~\ref{fig:profiling_res}); (ii) layer-wise channel reduction and expansion reduce slot utilization, leaving many ciphertext slots idle. Existing packing methods partially address these issues, by reducing either inner- or outer-rotations or increasing initial density, but rely on static, handcrafted layouts that degrade across layers~\cite{kim2021hear,ebel2025orion}. As channels shrink or feature maps evolve~\cite{iandola2016squeezenet,szegedy2015googlenet}, sparsely populated ciphertexts lead to proliferation of ciphertexts, poor SIMD efficiency, and high memory overhead. In short, current HE-CNN frameworks~\cite{aharoni2023helayers, ebel2025orion} neither jointly optimize rotation cost and slot utilization nor provide a principled manner to generate efficient layouts across diverse models and datasets.

\begin{figure}[t]
\centering
\includegraphics[width=0.85\columnwidth]{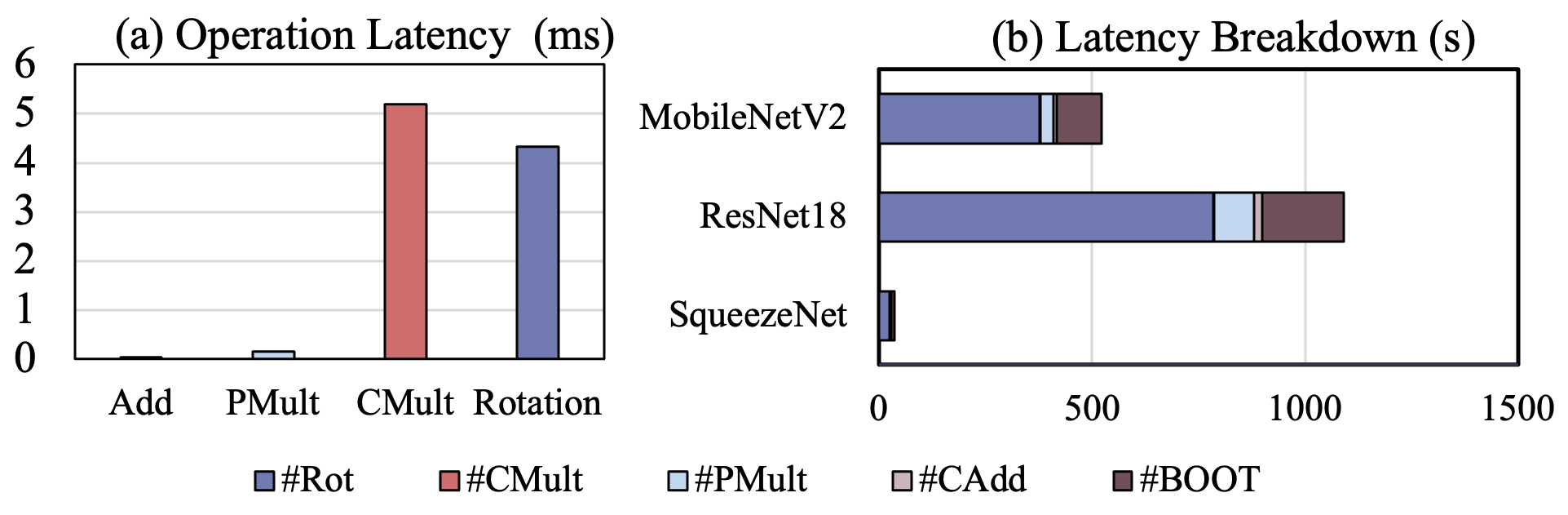}
% \vskip -0.10in
%\vspace{-3pt}
\caption{\small
(a) Latency comparison of HE primitives under encryption parameter $N=2^{16}$ on a GPU platform. 
(b) End-to-end encrypted inference latency breakdown of the SOTA Orion encoding~\cite{ebel2025orion} on ImageNet for three CNNs: SqueezeNet, ResNet18, and MobileNet. 
Note with ciphertext input \& plaintext model, \#CMult is very limited.  
%and appears only in polynomial activations.
}
%\caption{(a) HE primitves' operation latency comparison under encryption parameter $N=2^{16}$ in GPU environment. (b) Latency breakdown profiling results for the most recent SOTA method-Orion~\cite{ebel2025orion} on 3 benchmarks: SqueezeNet, ResNet18, MobileNetV2.
%}
%\textcolor{blue}{Add the motivation example 1. relationship between rotation number-memory usage-latency, 2. Wasted-slot ciphertext computation cost more memory and higher latency 3. higher-level ciphertext has higher latency}
\label{fig:profiling_res}
% \vskip -0.15in
%\vspace{-5pt}
\end{figure}

% We prototype \ourmodel on real systems using CPU- and GPU-based implementations, and evaluate our framework on MNIST, CIFAR10, and ImageNet-scale encrypted inference using LeNet, VGG5, SqueezeNet, ResNet18, and MobileNet. These workloads span the spectrum of representative small networks to deeper and larger HE-CNN models evaluated to date. The results show that across all benchmarks, \ourmodel can effectively reduce the rotation and keyswitch volume compared to state-of-the-art (SOTA) schemes while maintaining high ciphertext occupancy (i.e., CHET~\cite{dathathri2019chet} and Orion~\cite{ebel2025orion}). \ourmodel achieves up to $226.06\times$ and $228.83\times$ speedup, and reduces memory usage by up to 98.49\% and 75.6\% over the latest SOTA, \texttt{Orion}, on CPU- and GPU-based systems, respectively, for LeNet. For MobileNet, \ourmodel achieves up to $9.43\times$ and $4.55\times$ speedup, and reduces memory usage by up to 85.08\% and 75.68\% on CPU- and GPU-based systems, respectively.
%\ourmodel achieves up to $226.06\times$ and $228.83\times$ speedup and introduces up to 98.49\% and 75.6\% memory saving over the latest SOTA--Orion in CPU and GPU based systems, respectively. 

This paper addresses the performance-critical packing problem by introducing \ourmodel, a unified and automated framework that maximizes slot utilization and minimizes costly rotations for efficient and scalable HE-CNN inference across arbitrary CNN models, datasets, and batch configurations. The key insight is to treat ciphertext packing as an algorithmic degree of freedom that can be optimized using convolutional structure and layer-wise tensor geometry. To this end, \ourmodel consists of two complementary components. First, \frag provides a parameterized block decomposition of 4D feature tensors across width, height, channel, and batch dimensions, encoding each block into separate ciphertexts to decouple both adjacent intra-channel and cross-channel data dependencies in convolution. The optimal block size is derived from a convex model that balances inner- and outer-rotation costs, enabling efficient layouts across both small and large batch settings. Second, Architecture-aware Ciphertext Compression maintains high slot utilization across layers by consolidating sparsely filled ciphertexts and preventing fragmentation as the network evolves. Together, these two components allow \ourmodel to dynamically adapt to CNN layer structures and automatically generate ciphertext layouts that reduce rotation complexity while preserving high ciphertext occupancy. By reducing the number and cost of HE operations presented to hardware, \ourmodel amplifies the benefits of existing low-level HE accelerators, including optimized NTT units, key-switching circuits, and SIMD-aware engines~\cite{kim_bts_2022, fan_fastfhe_2025, kim_anaheim_2025, samardzic_f1_2021,yang_ola_2025}, without requiring any modification to model architectures, encryption parameters, or hardware. We prototype \ourmodel on both CPU- and GPU-based systems and evaluate it on encrypted inference workloads spanning MNIST, CIFAR10, and ImageNet-scale settings with LeNet, VGG5, SqueezeNet, ResNet18, and MobileNet. Across all benchmarks, \ourmodel reduces rotation and key-switch volume relative to prior SOTA schemes while maintaining high ciphertext occupancy. Compared with the latest SOTA, \texttt{Orion}, \ourmodel achieves up to $226.06\times$ and $228.83\times$ speedup and up to 98.49\% and 75.6\% memory reduction on CPU- and GPU-based systems, respectively, for LeNet; for MobileNet, it achieves up to $9.43\times$ and $4.55\times$ speedup and up to 85.08\% and 75.68\% memory reduction, respectively. To summarize, our main contributions are:
\begin{itemize}
    \item We emphasize ciphertext packing as a \emph{long-standing bottleneck} in HE-CNNs and formulate a principled, rotation-minimizing fragment layout with theoretical guarantees.
    \item We introduce a \emph{cross-layer ciphertext consolidation} mechanism that preserves high slot utilization. By mitigating slot waste introduced by layer operations (e.g., $1\times1$ convolutions for feature reduction), \multi reduces the number of ciphertexts needed for efficient SIMD processing in subsequent layers. 
    %and prevents ciphertext proliferation throughout the network.
    \item We develop the first \emph{unified and fully automated} HE-CNN packing framework that delivers efficient layouts for any model, dataset, or batch size, without manual tuning or runtime profiling. Its generality is shown by encompassing prior solutions as non-optimal special cases, while its optimality is both analytically and empirically validated.
    \item  We perform comprehensive experiments to evaluate \ourmodel in terms of throughput and memory efficiency. The results validate our key design choices, showing orders of magnitude speedup for HE-CNN inferences and up to $226.06\times$ and $109.96\times$ speedup and introduces up to 98.49\% and 75.6\% memory saving over Orion in CPU and GPU based systems, respectively demonstrating the effectiveness of our approach across platforms.
    %We demonstrate substantial end-to-end improvements on ImageNet-scale encrypted inference, showing that algorithmic workload shaping complements hardware acceleration and materially benefits CPUs, GPUs, and HE accelerators.
\end{itemize}

\color{black}

% \begin{figure}[t]
% \centering
% \includegraphics[width=\columnwidth]{figs/latency-breakdown.pdf}
% % \vskip -0.10in
% \caption{(a) Latency breakdown profiling result for a general 32-64 convolutional layer with 3x3 kernel.
% (b) HE operation latency comparison under encryption parameter $N=2^{15}$ used in our experiment.}
% %\textcolor{blue}{Add the motivation example 1. relationship between rotation number-memory usage-latency, 2. Wasted-slot ciphertext computation cost more memory and higher latency 3. higher-level ciphertext has higher latency}
% \label{fig:profiling_res}
% % \vskip -0.15in
% %\vspace{-20pt}
% \end{figure}

% \begin{figure}[t]
% \centering
% \includegraphics[width=0.9\columnwidth]{figs/motivation-example2.pdf}
% % \vskip -0.10in
% \caption{
% (a) Throughput (img/sec) comparison between CHET, Orion and our method on 3 benchmark models . (b) Amortized memory usage comparison under 3 methods on 3 benchmark models.
% (c) Rotation cost under 3 types of convolution operations with different packing.
% (d) Slot utilization under 3 different ciphertext size with different packing. }
% %\textcolor{red}{REMARK:require more detailed analysis, including more batched data could be processed in limited GPU memory capacity, and latency and memory compare between 3 methods showing the effect of different slot utilization }
% %\textcolor{red}{Remark: throughput~\# samples/min}}
% \label{fig:motivation_example}
% % \vskip -0.15in
% %\vspace{-20pt}
% \end{figure}
%\vspace{-4pt}
\section{Background}
%\vspace{-3pt}
\label{sec:pre}

\begin{table}[t]
\scriptsize
\centering
\caption{Notations table}
\label{tab:notations}
%\small
\footnotesize
\resizebox{\columnwidth}{!}{%
\begin{tabular}{ll}
\toprule
\textbf{Notation} & \textbf{Description} \\ 
\addlinespace[2pt]\midrule
$N$                                & Polynomial degree (number of coefficients) \\
$N/2$                              & Number of available slots in an encoded message \\
$\Delta$                           & Scale factor used for polynomial encoding \\
$Q$                              & Ciphertext modulus chain $\{q_0, q_1, \dots, q_L\}$ \\
%$L$                                & Length of the modulus chain (maximum multiplicative depth) \\
$\alpha$                           & Number of channels packed in one ciphertext \\
$BS$                               & Batch size of input sample\\
$K$                                & Convolution kernel size \\
$H,W$                              & Featuremap Height, Width\\
$N_{in}, N_{out}$                  & Input and Output channel number \\
\bottomrule
\end{tabular}%
}
%\vspace{-10pt}
\end{table}

In this section, we introduce the basics of \textit{Cheon--Kim--Kim--Song (CKKS)}, the encryption scheme used in this work. Table~\ref{tab:notations} summarizes the notation used throughout this paper.

%\textit{Homomorphic Encryption (HE)} enables computation directly on encrypted data without decryption. This work focuses on the state-of-the-art \textit{Cheon--Kim--Kim--Song (CKKS)} scheme~\cite{cheon2017homomorphic}, which is 
\vspace{1mm}
\textbf{CKKS} is a state-of-the-art homomorphic encryption (HE) scheme widely adopted for encrypted neural network inference due to its support for fixed-point real-number arithmetic~\cite{ han2020better, phan2024gpu, khan2023lossless}. A CKKS ciphertext represents a degree-$N$ polynomial in $\mathbb{Z}_q[X]/(X^N + 1)$ that encodes up to $N/2$ complex values, referred to as \textit{slots}--all processed in parallel for SIMD-style computation. CKKS supports several SIMD-based HE primitives essential for encrypted computation, including ciphertext (ct) addition $\textit{Add}(ct_1, ct_2) = ct_1 + ct_2$, ciphertext multiplication $\textit{CMult}(ct_1, ct_2) = ct_1 \circ ct_2$, plaintext--ciphertext multiplication $\textit{PMult}(ct, pt) = ct \circ pt$, slot rotation $\textit{Rot}(ct, k)$, which cyclically shifts encrypted vector elements with the offset $k$, and rescaling $\textit{Rescale}(ct, \Delta) = \nicefrac{ct}{2^\Delta}$, which manages noise after multiplications to prevent decryption failure, by dividing the ciphertext by $2^\Delta$ (or truncating $\Delta$ bits from its modulus), thereby consuming one ciphertext level.

\begin{figure*}[t!]
\centering
\includegraphics[width=1.0\textwidth]{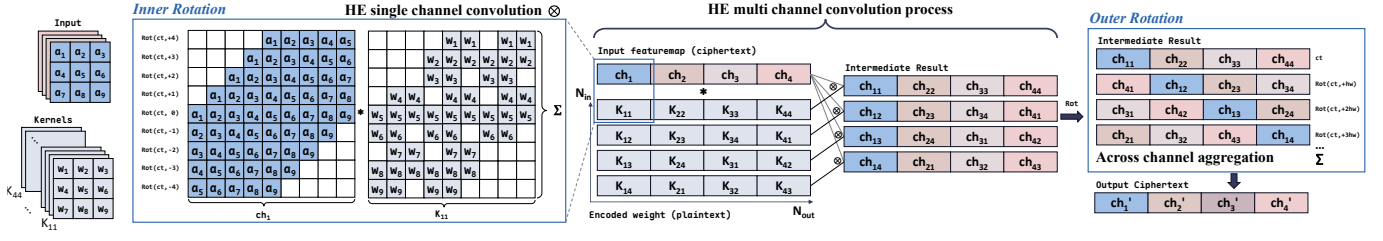}
%\vskip -0.05in
%\vspace{-15pt}
\caption{\small Illustration of HE multi-channel convolution with 4 input/output channels, $3\times3$ kernels, ciphertext inputs, and plaintext kernels.}
\label{fig:multi}
% \vskip -0.1in
%\vspace{-15pt}
\end{figure*}

\vspace{1mm}
\textbf{Rotation.} 
%Rotation, whose latency is close that of most expensive CMult, can be much higher than PMult and Add, e.g. 
Among HE primitives, \textit{Rotation}, together with \textit{CMult}, is substantially more expensive than \textit{PMult} or \textit{Add} (e.g., 4.8ms vs. 0.15ms), as shown in Fig.~\ref{fig:profiling_res} (a). This high cost arises from two components: an automorphism followed by a key-switching operation. The rotation is computed as:
\begin{equation}\label{rotation}
\vspace{-5pt}
\scriptsize
    \mathrm{Rot}(ct, k)
    = (c(X^{ik}), 0)
      + P^{-1}\!\big(a(X^{ik}) \cdot evk_{\text{rot}}^{k}\big),
\end{equation}
where $evk_{\text{rot}}^{k}$ is the rotation evaluation key with large modulus $Q$. For a ciphertext $ct = (c(X^i), a(X^i))$, the automorphism maps each coefficient index $i$ to $ik \bmod N$. The second term performs key switching, which dominates the latency, ensuring the output ct remains decryptable by the same secret key~\cite{samardzic2022craterlake}.

\color{black}

%In the realm of Multi-Party Computation (MPC)~\cite{ghodsi2020cryptonas,jha2021deepreduce,mishra2020delphi,lou2020autoprivacy}, there has been a concerted effort to mitigate the latency impact of non-linear operations, a critical factor in HE-inference, as exemplified by works that focus on reducing ReLU operations. 

% \begin{figure}[t]
% \centering
% \includegraphics[width=0.9\columnwidth]{figs/motivation-example2.pdf}
% % \vskip -0.10in
% \caption{
% (a) Throughput (img/sec) comparison between CHET, Orion and our method on 3 benchmark models . (b) Amortized memory usage comparison under 3 methods on 3 benchmark models.
% (c) Rotation cost under 3 types of convolution operations with different packing.
% (d) Slot utilization under 3 different ciphertext size with different packing. }
% %\textcolor{red}{REMARK:require more detailed analysis, including more batched data could be processed in limited GPU memory capacity, and latency and memory compare between 3 methods showing the effect of different slot utilization }
% %\textcolor{red}{Remark: throughput~\# samples/min}}
% \label{fig:motivation_example}
% % \vskip -0.15in
% %\vspace{-20pt}
% \end{figure}

\section{Motivation}
\label{sec:motivation}

\subsection{HE Multi-Channel Convolution} \label{sec:data_den}
We analyze how excessive rotations arise in multi-channel convolution, the dominant computation pattern in HE-CNN inference. Without loss of generality, we assume HE multi-channel convolution (multi-input, multi-output, MIMO) uses baby-step-giant-step (BSGS)~\cite{kim2021hear,intoci2023slytherin,ebel2025orion}, a common technique in SOTA HE inference~\cite{ebel2025orion,lee2022low} to decouple nested loops and reduce rotation complexity.

Given an input tensor $\mathbf{X} \in \mathbb{R}^{N_{\text{in}} \times H \times W}$, weight kernel $\mathbf{K} \in \mathbb{R}^{N_{\text{out}} \times N_{\text{in}} \times K \times K}$, and output $\mathbf{Y} \in \mathbb{R}^{N_{\text{out}} \times H \times W}$, the standard 2D convolution is:
\begin{equation}
\vspace{-5pt}
% \scriptsize
\footnotesize
\label{eq:multi-channel-conv}
Y_{n_{\text{out}}, h, w} =
\underbrace{\sum_{n_{\text{in}}}}_{\text{channel dependency}}
\;
\underbrace{\sum_{i,j} X_{n_{\text{in}}, h+i, w+j} \cdot K_{n_{\text{out}}, n_{\text{in}}, i, j}}_{\text{spatial dependency}}.
\end{equation}

Eq.~\ref{eq:multi-channel-conv} exhibits two nested dependencies: 
\textbf{spatial (neighboring pixels within each input channel)} and \textbf{channel (aggregation across input channels)}. %In HE convolution, both inherently require ciphertext rotations for alignment
In HE convolution, \textit{these dependencies inherently require ciphertext rotations} to align data for multiplication and accumulation (MAC), as illustrated in \figureautorefname~\ref{fig:multi}.
\begin{itemize}[noitemsep]
    \item \textbf{Inner rotations (spatial aggregation):} Each input ciphertext undergoes $(K^2-1)$ rotations to generate shifted copies for $K \times K$ convolution.
    \item \textbf{Outer rotations (channel aggregation):} After inner rotations, a ciphertext packing $\alpha = \lceil \frac{N}{2HW} \rceil$ channels must be aligned for channel-wise aggregation, requiring $(\alpha - 1)$ rotations per output ciphertext.
    %After inner rotations, packed channels $\alpha=\lceil N/(2HW) \rceil$ must be aligned for output aggregation, requiring $(\alpha-1)$ rotations per output ciphertext.
\end{itemize}

These rotations dominate runtime in large-scale SOTA HE-CNN inference, contributing $\sim70\%$ of total latency (including bootstrapping) for single-image ImageNet inference on MobileNet and ResNet (Fig.~\ref{fig:profiling_res}(b)).

\textbf{Computational perspective:} Optimizing HE convolution requires balancing inner and outer rotations while ensuring high throughput, i.e., \textit{producing ciphertexts that accommodate as many output channels as possible per computation}, and maintaining efficiency across subsequent layers. Inherent channel and spatial dependencies create complex rotation patterns, so naively packing multi-channel feature maps into a single ciphertext for SIMD parallelism is insufficient.

\textbf{Memory perspective:} SIMD throughput directly impacts memory usage. Low slot utilization (e.g., $50\%$) effectively doubles the number of ciphertexts, increasing memory footprint and inflating the computation required in downstream layers. This issue becomes prominent after feature or channel reduction followed by expansion layers (e.g., $1{\times}1$ convolutions) in modern CNNs such as MobileNet and ResNet.

\subsection{Limitations of SOTA Ciphertext Encodings}\label{sec:existing-ciphertext-packing-optim}
\label{Motiv:packing}

%\textbf{Existing HE Packing Methods.} 
Early HE-CNN frameworks, such as LoLa~\cite{brutzkus2019low}, introduce row-major ciphertext packing and store feature maps as 1D vectors to exploit SIMD parallelism. Later methods focus on
improving rotation efficiency or ciphertext utilization: 

\noindent\textbf{Inner-rotation optimization.} CHET~\cite{dathathri2019chet} and HElayers~\cite{aharoni2023helayers} reduce rotations for single-channel convolutions using pre-rotations, blocking, or batch packing. 

\noindent\textbf{Outer-rotation optimization.} Gazelle~\cite{juvekar2018gazelle}, Fast-HEAR~\cite{kim2022secure}, Multiplexed~\cite{lee2022low}, Orion~\cite{ebel2025orion}, and Hyena~\cite{singh2024hyena} extend packing
to multi-channel convolutions using interleaving and BSGS.
However, adjacent pixels often remain in the same ciphertext,
limiting rotation reduction.

\noindent\textbf{SIMD efficiency and dense packing.} Fast-HEAR, Multiplxed and Orion pack more channels to empty slots to improve
the throughput after operations like stride $\geq2$ convolutions and pooling. Fhelipe~\cite{krastev2024tensor} merges sparsely filled slots post-layer but ignores the next layer’s computation pattern, yielding suboptimal packing.

Table~\ref{tbl:compare_packing} provides a high-level comparison between prior methods and our work. 
While some methods achieve dense packing initially, few are able to preserve this density after channel or feature reduction, and none fully optimize rotation overhead. Overall, current HE-CNN packing methods exhibit three key limitations.
%While some achieve dense packing initially, few maintain it after channel or feature reduction, and they all fail to fully optimize rotations. Overall, all  existing HE-CNN packing methods exhibit three key limitations.
%As a result, existing HE-CNN packing methods exhibit three key limitations. 
\ding{202} \textit{Static and heuristic designs}: Packings are manually crafted, layer-agnostic, and lack principled guidance. \ding{203} \textit{Cross-layer fragmentation}: Initially dense packings degrade after channel/feature reduction, yielding poor slot utilization and increased memory overhead. \ding{204} \textit{Incomplete rotation optimization:} Prior works typically reduce either inner or outer rotations, but rarely addresses both jointly across layers.
%\begin{itemize}[noitemsep]
% \item \textbf{Static and heuristic designs:} Packings are manually crafted, layer-agnostic, and lack principled guidance.
% \item \textbf{Cross-layer fragmentation:} Initially dense packings degrade after channel/feature reduction, yielding poor slot utilization and increased memory overhead.
% \item \textbf{Incomplete rotation optimization:} Prior work typically reduces either inner or outer rotations, but rarely addresses both jointly across layers.
% \end{itemize}
Overall, these limitations point to a broader gap: prior methods cannot jointly optimize rotation cost and slot utilization within a unified, automated, principled framework (as achieved by \ourmodel in Table~\ref{tbl:compare_packing}).
%\end{greenbox}

\begin{table}[b!]
% \vspace{-10pt}
\caption{SOTA HE packing methods comparison.}
% \vspace{-10pt}
\label{tbl:compare_packing}
% \vskip 0.15in
\centering
\resizebox{\columnwidth}{!}{
\begin{tabular}{c|ccc|cc}
 & \multicolumn{3}{c|}{Densely-Packed} & \multicolumn{2}{c}{Rot Optimization} \\
\multirow{-2}{*}{Model} & at initial & after channel-reduce & after feature-reduce & multi-channel & single-channel \\ \hline
CryptoNets~\cite{gilad2016cryptonets} & \XSolidBrush & \XSolidBrush & \XSolidBrush & \XSolidBrush & \XSolidBrush \\
CHET~\cite{dathathri2019chet} & \XSolidBrush & \XSolidBrush & \XSolidBrush & \XSolidBrush & \XSolidBrush \\
HELayers~\cite{aharoni2023helayers} & \XSolidBrush & \XSolidBrush & \Checkmark & \XSolidBrush & \Checkmark \\
Multiplxed~\cite{lee2022low} & \Checkmark & \XSolidBrush & \Checkmark & \Checkmark & \XSolidBrush \\
Fhelipe~\cite{krastev2024tensor} & \Checkmark & \XSolidBrush & \Checkmark & \Checkmark & \XSolidBrush \\
%\rowcolor{blue!10}
Hyena~\cite{singh2024hyena} & \Checkmark & \XSolidBrush & \XSolidBrush & \Checkmark& \XSolidBrush \\
Orion~\cite{ebel2025orion} & \Checkmark & \XSolidBrush & \Checkmark & \Checkmark& \XSolidBrush \\
\rowcolor[HTML]{808080}
% \rowcolor{LightGreen}
\textbf{FEnc$^2$} & \Checkmark & \Checkmark & \Checkmark & \Checkmark & \Checkmark \\ \hline
\end{tabular}%
}
% \vspace{-0.20in}
\end{table}

\begin{figure*}[t!]
\centering
\includegraphics[width=0.8\textwidth]{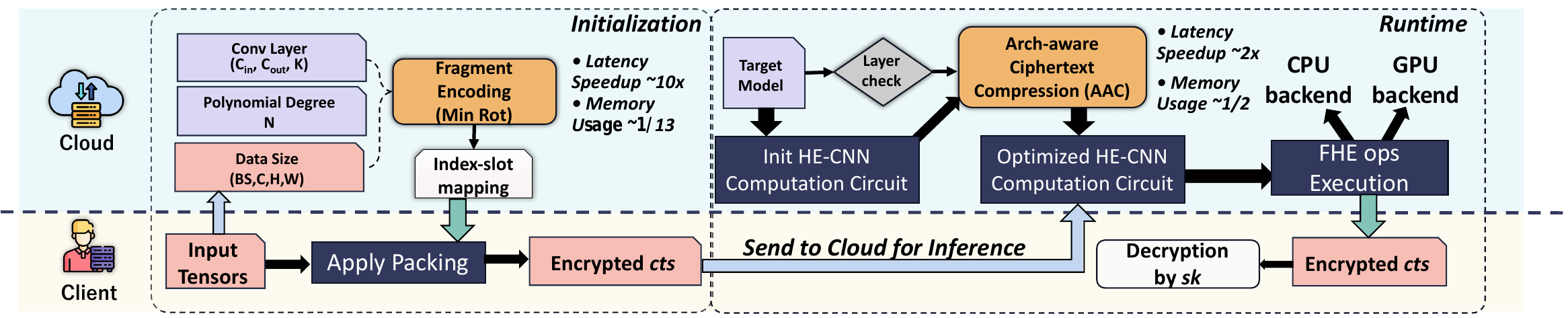}
%\vskip -0.10in
\small
\caption{An overview of \ourmodel, which includes two components: 1) Conv-aware HE fragment encoding for outputting an optimal block size selection; 2) Arch-ware Ciphertext Compression for adapting to channel dimension change.} 
\label{fig:com}
% \vskip -0.1in
%\vspace{-15pt}
\end{figure*}

\subsection{Insights and Positioning of Our Work--\ourmodel}
\ourmodel addresses these gaps by providing a unified, automated, and layer-aware framework for HE-CNN packing. Its parameterized block-size model partitions feature maps across ciphertexts, decoupling adjacent-pixel dependencies within each block to minimize inner rotations, and balances inner- and outer-rotation costs through a principled convex optimization. To maintain high throughput across layers, \ourmodel dynamically applies ciphertext compression and merging, maximizing slot utilization for subsequent layer computations without altering the packing format. This ensures that each layer can process as many output channels per ciphertext as possible, reducing the total number of ciphertexts while preserving SIMD efficiency. Unlike prior static or single-layer heuristics, \ourmodel automatically adapts to layer-wise computation dependencies, producing ciphertext layouts with provable guarantees on both rotation complexity and slot utilization. By delivering efficient, architecture-aware data layouts to the runtime, \ourmodel reduces the overall HE workload exposed to hardware and complements low-level optimizations (e.g., NTT, key-switching), enabling end-to-end HE-CNN acceleration on any hardware platform.

\section{\ourmodel Framework Design}

\subsection{Design Overview}

Building on the aforementioned insights, \ourmodel provides a unified and automated framework that jointly addresses rotation minimization and slot utilization: two criteria that prior packing strategies optimize only in isolation. Rather than optimizing only the initial input packing or performing ad-hoc reorganization after each layer, \ourmodel leverages application-level structure (input layout, convolution patterns, and network topology) to generate packing schemes that remain efficient throughout end-to-end CNN inference.

At a high level, \ourmodel produces data layouts that (i) minimize rotation operations required for multichannel convolution, and (ii) maintain high slot utilization across layers even under feature expansion or reduction. This enables consistent SIMD throughput and controlled ciphertext growth, improving both latency and memory efficiency on HE accelerators.

\ourmodel consists of two complementary components:

\begin{itemize}[noitemsep]
    \item \textbf{\frag} (\S\ref{sec:conv-aware-he-fragment-encoding}): a convolution-aware fragment encoding that theoretically minimizes rotation cost by partitioning and packing features into independent block-wise ciphertexts. Its generality and optimality are analyzed in \S\ref{sec:rot-complexity-cases}. Table~\ref{matrixcompare} also demonstrates its provable advantages over prior packing strategies in rotation complexity.
    \item \textbf{\multi} (\S\ref{sec:channel}): a cross-layer slot-utilization optimizer that densifies sparsely filled ciphertexts via rotation-mask-add and scale-bit--aware ciphertext compression, enabling efficient processing after channel/feature expansion or reduction while preserving the packing format.
    %a cross-layer slot-utilization optimizer that densifies sparsely filled ciphertexts via rotation–mask–add and scale bit allocation for ciphertext compression, enabling efficient handling of channel/feature expansion or reduction without altering the underlying packing format.
\end{itemize}

Together, they form a coherent pipeline that sustains optimal rotation complexity and high slot utilization across the entire network, delivering robust encrypted inference performance. 
%for HE-CNN inference.

% Our unified data packing framework, namely \ourmodel, is a design automation tool for accelerating end-to-end HE-CNN inference  with \textit{awareness of both the computation pattern (e.g., convolution) and the neural network architecture} through principal guarantee of minimal rotation cost and maximum slot utilization both at the initial data packing stage and throughout layer-wise processing to fully leverage limited SIMD parallel computing units and memory capacity on hardware. It consists of two key components: (1) \frag, introduced in \ref{sec:conv-aware-he-fragment-encoding}, which minimizes rotation costs with theoretical guarantees by structurally packing convolution-relevant features into independent ciphertexts in a block-wise manner; we analyze multiple scenarios in \sectionautorefname~\ref{sec:rot-complexity-cases} to demonstrate the generalizability of this approach. Then, (2) \multi, introduced in \ref{sec:channel}, maintains high SIMD efficiency across CNN layers by compressing ciphertexts with empty slots into densely packed ciphertexts in a convolution-friendly layout. Together, \ourmodel sustains high efficiency and throughput throughout the entire computation. 

%Moreover, to mitigate the additional rotations and HE multiplicative level consumption introduced in this process, we develop large-offset ct rotation and ct level-aware quantization techniques, which align the underlying HE optimizations with CNN inference.

\textbf{How \ourmodel is Used in Practice?}
We consider a standard cloud-based inference setting where the client holds \textbf{private input data} and the server hosts \textbf{pretrained plaintext CNN models}. As Fig.~\ref{fig:com} shows, during initialization the client sends only non-sensitive metadata, input dimensions $(H, W, C)$, batch size $BS$, and the model identifier. This information reflects only tensor shapes and reveals no semantic information about the client’s actual input.

Using this metadata and the plaintext model, \ourmodel automatically determines (i) the optimal block size $S$ for \frag and (ii) the model specific layer-wise intermediate ciphertext compression strategy for \multi. %producing an optimized HE-CNN circuit tailored to the model. 
These decisions depend solely on tensor dimensions and convolutional structure. \ourmodel then returns an optimal index-slot mapping to the client.

The client locally performs CKKS encoding and encryption according to this mapping and sends \emph{only ciphertexts} to the server. The client never observes model weights or intermediate activations, and the server never receives plaintext inputs. At runtime, the server executes the optimized HE-CNN circuit entirely over ciphertexts using CPUs, GPUs, or custom HE accelerators, without any data-dependent profiling or adaptation. \textit{Because \ourmodel relies only on public metadata and requires neither architecture modifications no runtime profiling, it provides a drop-in acceleration path for existing MLaaS deployments while preserving the standard CKKS semantic security guarantees.}

\textbf{Why \ourmodel is Hardware-Agnostic Yet Architecturally Transformative?}
\ourmodel operates entirely at the application level, using convolutional structure and model topology to reduce HE operation counts. Because its packing strategies depend only on tensor dimensions and network structure, not on GPU organization, memory hierarchy, or ASIC microarchitecture. \ourmodel is compatible with any CPU-, GPU-, FPGA-, or ASIC-based HE accelerator. 
%At the same time, by substantially reducing ciphertext count, rotations, keyswitches, and NTTs (Table~\ref{tbl:detail-count-compare}), it reshapes the computation and communication workload on hardware, delivering system-level gains without hardware changes.
At the same time, by substantially lowering ciphertext count, rotations, keyswitches, and NTTs (See Table~\ref{tbl:detail-count-compare}), it reshapes the computation and communication demands placed on hardware, yielding system-level benefits without requiring hardware changes.

\begin{figure*}[t!]
\centering
\includegraphics[width=1.02\textwidth]{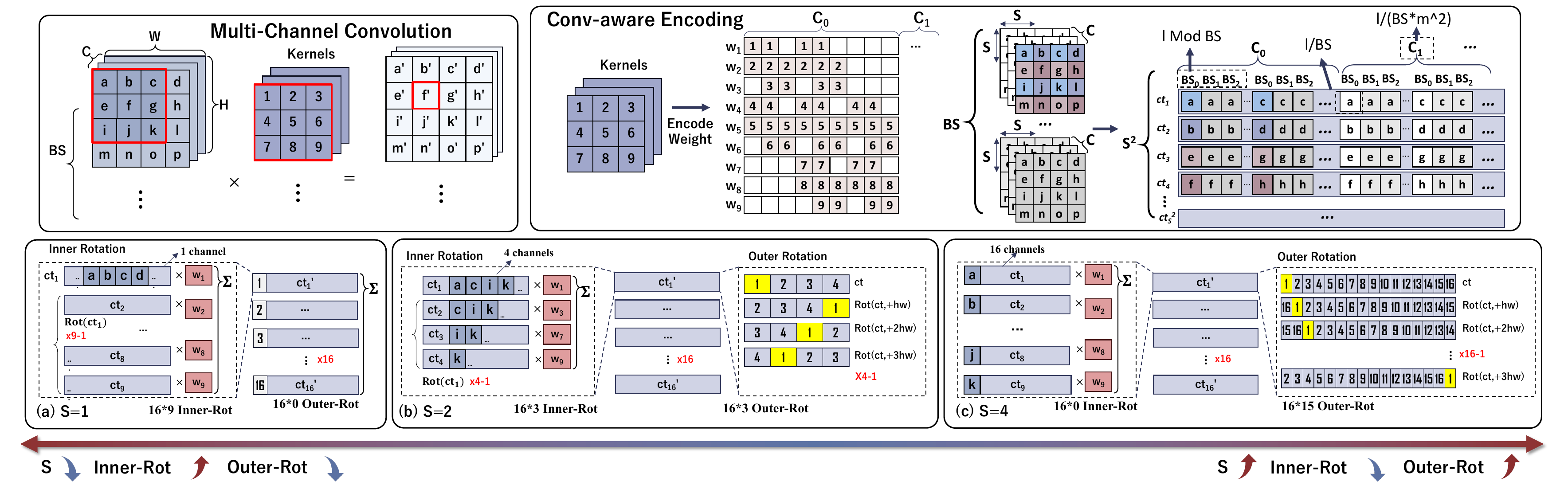}
%\vskip -0.10in
%\vspace{-15pt}
\small
% \caption{Conceptual view of \frag illustrating its generality and optimality using an example input $1{\times}16{\times}4{\times}4$ and convolution $(16,16,3,1)$ ($BS{=}1$). All ciphertexts are densely packed (16 slots). (a) $S{=}1$: non-optimal, zero outer rotations (row-major/CHET, Orion-style\cite{dathathri2019chet,kim2021hear,ebel2025orion}). (b) $S{=}2$: \textbf{ours-optimal, minimal rotation cost by balancing inner and outer rotations (equal inner/outer rotations)}. (c) $S{=}4$: non-optimal, zero inner rotations with max outer-rotations (CryptoNets-style\cite{gilad2016cryptonets}).}
\caption{Conceptual illustration of \frag demonstrating its generality and optimality for input $1{\times}16{\times}4{\times}4$ and convolution $(16,16,3,1)$ ($BS{=}1$). All ciphertexts are fully packed (16 slots). (a) $S{=}1$: non-optimal, no outer rotations (row-major/CHET, Orion-style~\cite{dathathri2019chet,kim2021hear,ebel2025orion}). (b) $S{=}2$: \textbf{optimal (ours)}, minimizing rotation cost by balancing inner and outer rotations. (c) $S{=}4$: non-optimal, no inner rotations and maximal outer rotations (CryptoNets-style~\cite{gilad2016cryptonets}).}
\label{fig:2x2}
%\vskip -0.10in
%\vspace{-10pt}
\end{figure*}

\textbf{Complementing Low-Level Primitive Acceleration.}
Prior HE accelerators typically optimize individual primitives such as NTT, keyswitching, or bootstrapping~\cite{kim_bts_2022 , jung_over_2021, kim_ark_2022, jiao_neo_2025}. These approaches improve per-kernel efficiency but cannot affect the number of primitives dictated by the HE-CNN computation graph. In contrast, \ourmodel reduces the \emph{number} of required rotations, keyswitches, and NTTs across the network. This reduction is complementary to low-level acceleration: it lowers the workload presented to hardware, amplifying the effects of optimized primitives and enabling larger end-to-end gains than primitive-level improvements alone.

Overall, \ourmodel provides an algorithmically driven architectural optimization that coexists naturally with existing accelerators, improving hardware utilization by structurally reducing HE workload at the source.

\begin{algorithm}[htbp]
\caption{\frag}
\scriptsize
\label{alg:fragment-packing}
Given input feature maps $X \in \mathbb{R}^{BS \times C \times H \times W}$ and block size $S$:
\begin{enumerate}[leftmargin=*, itemsep=2pt, parsep=0pt, topsep=2pt]
    \item \textbf{Step 1: Square padding.}
    Compute $M = \max(\mathrm{pad}(H), \mathrm{pad}(W))$ and zero-pad $X$ to shape $(BS, C, M, M)$.

    \item \textbf{Step 2: Block partition.}
    Set $m = M / S$ and partition each $M \times M$ feature map into an $m \times m$ grid of $S \times S$ blocks.

    \item \textbf{Step 3: Fragment packing and encryption.}
    For each intra-block coordinate $(u,v)$ with $0 \le u,v < S$:
    \begin{enumerate}[itemsep=1pt, parsep=0pt, topsep=1pt]
        \item Collect element $(u,v)$ from all $m \times m$ blocks to form  
              $X_{(u,v)} \in \mathbb{R}^{C \times BS \times m^{2}}$.
        % \item Flatten $X_{(u,v)}$ into a 1D vector of size $\frac{N}{2}$ by assigning each element  
        %       $X^{(u,v)}_{ijk}$ to slot $l$ according to
        %       \[
        %           X^{(u,v)}_{ijk} \;\longrightarrow\; \text{slot } l,
        %       \]
        %       where
        %       \[
        %           i = \left\lceil \frac{l}{BS \cdot m^{2}} \right\rceil,\qquad
        %           j = l \bmod BS,\qquad
        %           k = \left\lceil \frac{l}{BS} \right\rceil \bmod m^{2}.
        %       \]
\item Flatten $X_{(u,v)}$ into a 1D vector of size $\frac{N}{2}$ by assigning each element  
      $X^{(u,v)}_{ijk}$ to slot $l$ according to
        \[
        \begin{alignedat}{2}
        & X^{(u,v)}_{ijk} &\;\longrightarrow\;& \text{slot } l,\\
        & \text{where } i &=& \left\lceil \tfrac{l}{BS \cdot m^{2}} \right\rceil,\quad
        j = l \bmod BS,\quad
        k = \left\lceil \tfrac{l}{BS} \right\rceil \bmod m^{2}.
        \end{alignedat}
        \]
    \end{enumerate}

    \item \textbf{Step 4: Return the fully packed ciphertexts $\{ct_{(u,v)}\}_{u,v}$.}
          %Return the fully packed ciphertexts $\{ct_{(u,v)}\}_{u,v}$.
\end{enumerate}
% \vspace{-15pt}
\end{algorithm}

\subsection{\frag}
\label{sec:conv-aware-he-fragment-encoding}

As discussed in Section~\ref{sec:motivation}, an optimal packing strategy must jointly minimize both inner and outer rotations; optimizing only one, as in prior studies, inevitably lead to suboptimal performance. Achieving such a global optimum requires coupling these two rotation costs within a rigorous, tractable formulation-an aspect not addressed in existing methods. \frag addresses this challenge by introducing a 4D subblock-style data layout that deliberately decouples spatial dependencies according to convolutional kernel structure. By effectively leveraging this 4D structure-spanning batch, channel, and the feature map’s height and width, and carefully tuning the block size, it balances inner- and outer-rotation demands, enabling principled and near-optimal rotation complexity for HE convolution.

\subsubsection{Unified Encoding Formulation} 

To achieve a unified encoding formulation that produces fully packed ciphertexts while maintaining flexibility from the initial input stage, we represent the data as a 4D tensor $X \in \mathbb{R}^{BS \times C \times H \times W}$, corresponding to batch size, channels, and feature-map height and width. Algorithm~\ref{alg:fragment-packing} summarizes the packing procedure: 

\textbf{Step 1}: zero-pad the input to a square size $M = \max(\mathrm{pad}(H), \mathrm{pad}(W))$;  
\textbf{Step 2}: partition each feature map into an $m \times m$ grid of $S \times S$ blocks;  
\textbf{Step 3}: for each intra-block position $(u,v)$, gather all elements at that position across all blocks into a tensor $X_{(u,v)}$, flatten it, and assign elements to ciphertext slots according to the mapping in Algorithm~\ref{alg:fragment-packing};  
\textbf{Step 4}: the fully-packed ciphertexts $ct_{(u,v)}$ are generated.

Fig.~\ref{fig:2x2} illustrates the encoding concept. For a $2 \times 2$ block interacting with a $3 \times 3$ kernel, neighboring pixels (e.g., $a, b, e, f$) are assigned to the same slot indices across distinct ciphertexts. Remaining slots are filled first with pixels from other samples (for batching) and then from other channels, if available. By tuning the block size $S$, \frag balances inner- and outer-rotation costs: smaller $S$ reduces outer rotations by packing less channels together, while larger $S$ reduces inner rotations by distributing pixels from the same kernel window across different ciphertexts. This principled tradeoff yields near-optimal rotation complexity for HE convolution (Section~\ref{sec:optimal}).

\textbf{Generality:} \frag provides a unified 4D-aware framework that encompasses prior HE packing schemes as special cases (non-optimal) and generalizes them. \textbf{Notably, $S=1, BS=1$ recovers row-major encoding and its latest variants (e.g., Orion)~\cite{kim2021hear,dathathri2019chet,ebel2025orion}, while $S=M$ reduces to pixel-wise encoding (CryptoNets)~\cite{gilad2016cryptonets}.} Illustrative examples in Fig.~\ref{fig:2x2} further demonstrate the impact of block size: (a) $S=1$ (row-major) yields zero outer rotations, (b) $S=2$ (\textbf{ours-optimal}) achieves minimal rotation cost by balancing inner and outer rotations, and (c) $S=4$ (pixel-wise) eliminates inner rotations but increases outer rotations.

%\textbf{Generality} \frag provides a unified 4D-aware framework that covers prior HE packing schemes as special cases and generalizes them: \textbf{Special cases (non-optimal): $S=1, BS=1$ recovers row-major encoding and its latest variants--Orion~\cite{kim2021hear,dathathri2019chet,ebel2025orion}, while $S=M$ reduces to pixel-wise encoding--CryptoNets~\cite{gilad2016cryptonets}.} As illustrative examples in Fig~\ref{fig:2x2} (a) (b) (c) further shows: $S=1$ (row-major) yields zero outer rotations, $S=2$ (\textbf{ours-optimal}) achieves minimal rotation cost by balancing inner and outer rotations, and $S=4$ (pixel-wise) eliminates inner rotations but increases outer rotations.  

This demonstrates that \frag not only unifies and generalizes prior SOTA methods but also enables systematic, near-optimal HE-CNN inference through block-size tuning and 4D layout exploitation.

\subsubsection{Analytical Model of Rotation Complexity}\label{sec:rot-complexity-cases}
%As illustrated in Figure~\ref{fig:2x2}, the total rotation complexity in our encoding is closely tied to the choice of block size $S$ for a given kernel size $K$. 
%Since $S$ determines the total rotation complexity, To determine the optimal block size $S$, we first analyze the inner rotation complexity for single-channel convolution, followed by the outer rotation complexity for multi-channel convolution, and finally we assess the total rotation complexity.
%In this section, we create a analyatical model
Since the total rotation complexity depends critically on the block size $S$, we next present an analytical model to quantify and optimize the overall rotation complexity.

\noindent\textbf{Inner-rotation complexity.} 
Assuming full utilization of ciphertext slots at the initial packing stage, let the number of ciphertexts be $\textstyle N_\text{in}/\alpha$. We categorize the inner-rotation complexity based on the relationship between block size $S$ and kernel size $K$:

\begin{itemize}[nolistsep,leftmargin=*]
    \item \textbf{CASE 1: $K > S$}. Adjacent sub-blocks overlap (Fig.~\ref{fig:2x2} (a),(b)). In the worst case, each ciphertext requires 
    $\textstyle (\lceil K/S \rceil^2 - 1)$ rotations to compute single-channel convolution results, compared to $O(K^2-1)$ in prior methods.

    \item \textbf{CASE 2: $K \leq S < M$}. The convolution kernel spans at most one block, so inner rotations are generally unnecessary. Only edge computations involve up to $\textstyle 4(S-1)$ rotations.

    \item \textbf{CASE 3: $S = M$}. When the block size equals the feature map size (Fig.~\ref{fig:2x2} (c)), the encoding reduces to pixel-wise encoding~\cite{gilad2016cryptonets}, where each pixel occupies a separate ciphertext. No inner rotations are required.
\end{itemize}

Hence the inner-rotation complexity can be expressed as:
%$Rot_{inner}$
\begin{equation}
\label{equ:inner}
% \footnotesize
\scriptsize
    Rot_{inner}\textit{=} \left\{
    \begin{aligned}
    %& 2 \cdot (\lceil K/S \rceil)^2-1) \cdot K^2, K > S\\
    %& 2 \cdot 4 \cdot (S - 1), K \leq S \\
    & \frac{N_{in}}{\alpha} \times (\lceil K/S \rceil)^2-1), K > S\\
    & \frac{N_{in}}{\alpha}\times \nicefrac{4(S - 1)}{S^2}, K \leq S < M\\
    & 0, S=M \\
    \end{aligned}
    \right.
    % \vspace{-5pt}
\end{equation}
\noindent\textbf{Outer-rotation complexity.} 
For inter-channel convolution, consider a convolutional layer with dimensions $(N_\text{in}, N_\text{out}, K)$. Given an input batch size $BS$ and a chosen block size $S$, each ciphertext packs $\frac{\alpha S^2}{BS}$ channels. Based on Figure~\ref{fig:multi}, the corresponding outer-rotation complexity can be expressed as:
\begin{equation}\label{eq:outer-rot}
\scriptsize
        Rot_{outer}\textit{=} \frac{N_{out}}{\alpha} \times (\frac{\alpha S^2}{BS}-1)
\end{equation}
    %(\frac{N \cdot S^2}{2 \cdot M^2 \cdot BS}-1) \frac{O \cdot 2 \cdot M^2 \cdot BS}{N \cdot S^2}

\noindent\textbf{Total rotation complexity.}  The overall rotation complexity can be expressed by combining inner- and outer-rotation contributions. We consider two primary scenarios based on the block size $S$ and kernel coverage:

% when $K \leq S$:
% \begin{equation}\label{eq:total_rot}
% \scriptsize
%      Rot_{total} = \frac{BS}{\alpha} ( \frac{4(S-1)}{S^2} \cdot N_{in} + (\frac{\alpha S^2}{BS}-1) \cdot N_{out}) \\
% \end{equation}

\begin{equation}
\label{eq:total-rot}
% \footnotesize
\scriptsize
% \vspace{-5pt}
\setlength\abovedisplayskip{0pt}
    Rot_{total}\textit{=} \left\{
    \begin{aligned}
    %& 2 \cdot (\lceil K/S \rceil)^2-1) \cdot K^2, K > S\\
    %& 2 \cdot 4 \cdot (S - 1), K \leq S \\
    & \frac{BS}{\alpha} ( \frac{4(S-1)}{S^2} N_{in} + (\frac{\alpha S^2}{BS}-1) N_{out}), K \leq S < M\\
    & \frac{BS}{\alpha} ( (\frac{K^2}{S^2}-1) N_{in} + (\frac{\alpha S^2}{BS}-1) {N_{out})}, K > S\\
    \end{aligned}
    \right.
    % \vspace{-5pt}
\end{equation}

% \begin{align*}
% %    &\frac{I \cdot 2\cdot M^2 \cdot BS}{N \cdot S^4} 4(S-1) + (\frac{N 
% %    \cdot S^2}{2 \cdot M^2 \cdot BS}-1) \frac{O \cdot 2 \cdot M^2 \cdot BS}{N \cdot S^2} \\
% %&=O+\frac{2\cdot M^2 \cdot BS}{N} (\frac{I \cdot 4(S-1)}{S^4}-\frac{O}{S^2})
%     Rot_{total}\textit{=}&\frac{I \cdot 2\cdot M^2 \cdot BS}{N \cdot S^2} 4(S-1) + (\frac{N 
%     \cdot S^2}{2 \cdot M^2 \cdot BS}-1) \frac{O \cdot 2 \cdot M^2 \cdot BS}{N} \\
%      &=\frac{2\cdot M^2 \cdot BS }{N} (\frac{I \cdot 4(S-1)}{S^2}-O \cdot (\frac{N 
%     \cdot S^2}{2 \cdot M^2 \cdot BS}-1))
% \end{align*}
% When $K > S$, 
% \begin{equation}\label{eq:total_rot_2}
% \scriptsize
% %    &\frac{I \cdot 2\cdot M^2 \cdot BS}{N \cdot S^2} ((\nicefrac{K}{S})^2-1) + (\frac{N 
% %    \cdot S^2}{2 \cdot M^2 \cdot BS}-1) \frac{O \cdot 2 \cdot M^2 \cdot BS}{N \cdot S^2} \\
% %    &=O+\frac{2\cdot M^2 \cdot BS}{N} (\frac{I \cdot K^2}{S^4}-\frac{O}{S^2})
%     % Rot_{total}\textit{=}&\frac{I \cdot 2\cdot M^2 \cdot BS}{N } (\frac{K^2}{S^2}-1) + (\frac{N 
%     % \cdot S^2}{2 \cdot M^2 \cdot BS}-1) \frac{O \cdot 2 \cdot M^2 }{N} \\
%     % &=\frac{2\cdot M^2 \cdot BS}{N} (I \cdot (\frac{\cdot K^2}{S^2}-1)-O \cdot (\frac{N 
%     % \cdot S^2}{2 \cdot M^2 \cdot BS}-1) )
%     %Rot_{outer}\textit{=} (\frac{N \cdot S^2}{2 \cdot M^2 \cdot BS}-1) \frac{O \cdot 2 \cdot M^2 \cdot BS}{N}
%     Rot_{total} = \frac{BS}{\alpha} ( (\frac{K^2}{S^2}-1) N_{in} + (\frac{\alpha S^2}{BS}-1) {N_{out})} \\
% \end{equation}

\noindent\textbf{Special Case: Large Batch.}  
For batch inference with a sufficiently large $BS$, different inputs within the batch can be packed into ciphertexts instead of packing multichannel data from a single input. This completely eliminates the need for outer rotations, leaving only the inner-rotation cost:
\begin{equation}
\label{equ:large}
\scriptsize
    Rot_{\text{amortized}}= \begin{cases}
         \frac{4(S-1)}{S^2} \frac{N_{in}}{\alpha}, & K \leq S < M\\[2mm]
        (\lceil \frac{K}{S} \rceil ^2-1) \frac{N_{in}}{\alpha}, & S < K
    \end{cases}
\end{equation}
As $S$ increases, the amortized rotation complexity (a.k.a rotation per sample) $Rot_{\text{amortized}}$ decreases. We experimentally validate this special case in Section~\ref{sec:detailed-ab-study}.

\subsubsection{Theoretical Foundation for Minimal Rotation Cost} 
\label{sec:optimal}
%Based on above analysis, increasing the block size $S$ reduces the inner-rotation cost, but excessively large $S$ can waste ciphertext slots, e.g. $S^2 > N_{ct}$. 
%\subsubsection{Optimal Block Size with Theoretical Guarantee}
To determine the optimal block size $S$ that minimizes rotation cost, we first constrain $S$ to prevent inefficient ciphertext slot utilization in \frag. Excessively large $S$ with insufficient data (e.g., $S^2 > \frac{BS \cdot N_{in}}{\alpha}$) leaves many slots empty, wasting limited memory and SIMD parallel units. For example, a $4\times4$ single-channel feature map with $S=2$ produces 4 densely packed ciphertexts, each with 4 slots, whereas $S=4$ results in 16 ciphertexts with only one pixel each, causing a 3$\times$ increase in Multiply-Accumulate (MAC) operations and 4$\times$ memory overhead due to only 25\% slot utilization. To prevent this, $S$ is upper-bounded by
\begin{equation}\label{equ:opt}
\scriptsize
S \leq \sqrt{ \frac{BS\cdot N_{in}}{\alpha} }.
\end{equation}

Within this bound, enlarging $S$ reduces the inner-rotation complexity by $\frac{1}{S^2}$ but increases the number of channels per ciphertext to $S^2\alpha$, raising outer-rotation overhead (Equation~\ref{eq:outer-rot}). This trade-off defines the optimal block size $S^*$.

\begin{table}[t]
  \centering
  % 1. 将 caption 移出 tcolorbox 以确保交叉引用和列表生成正常
  \caption{Analytical Complexity Comparison (v.s. SOTAs) %\textcolor{red}{From Fan: only highlight the changes made. I belive only the Hyena+ row?}
  }
  \label{matrixcompare}
  %\vspace{-5pt}
  
  % \begin{tcolorbox}[
  %     enhanced,
  %     colback=green!10, 
  %     colframe=none, 
  %     sharp corners, 
  %     boxrule=0pt,
  %     left=2pt, right=2pt, top=5pt, bottom=5pt, % 调整边距
  %     width=\columnwidth, % 不要用 hbox，因为内部 resizebox 强制要求整栏宽
  %     center,
  %     nobeforeafter
  % ]
    \centering % 确保内容在盒子内居中
    \footnotesize % 建议在 resizebox 外部统一样式
    \resizebox{0.95\linewidth}{!}{ % 使用 0.95\linewidth 给盒子边距留出空间
      \begin{tabular}{llll}
        \toprule
        Metrics & Description & Amortized Complexity \\
        \midrule
        \multirow{6}{*}{Rot}   
          & CHET      & $ O(K^2 \cdot \alpha) $ \\
          & HElayers  & $ O( (K^2+K(M/t_1+M/t_2)) \cdot \alpha/BS ) $ \\
          & Orion     & $ O ( K^2+K\cdot \alpha ) $ \\
          %\rowcolor{blue!10}
          & Hyena+    & $ O ( min( (\alpha/BS)^2 ,1) K^4+ min(\alpha/BS,1) \cdot K^2)) $ \\
          % \rowcolor[HTML]{EFEFEF}
          %\rowcolor{blue!10}
          & Batchwise+ & \textbf{ $ O ( max( \frac{\alpha M^2}{BS}, N_{out}) \lceil \nicefrac{M^2}{BS} \rceil )$ } \\
          & \textbf{Ours-optimal} & \textbf{ $ O ( \sqrt \frac{\alpha \cdot K^2}{BS} ) $ } \\
        \bottomrule
      \end{tabular}
    }
  %\end{tcolorbox}
  \parbox{\textwidth}{\vspace{1ex}\tiny\textit{Table notes}: 
  Hyena+ prioritizes the batch dimension over the input channel dimension ($\alpha$) during packing, whereas \\ HELayers and \ourmodel prioritize the batch dimension over both input pixels and input channels. Batchwise+ is much \\ more complex than Hyena+, since the feature map size $M^2 \gg K^2, K^4$ (e.g., $M=32$, $K=3$).}
  %Hyena+ prioritizes batch dimension over input channel dimension-$\alpha$, while HElayers and \ourmodel, prioritize \\the batch dimension over input pixels and input channels. }}
  %\vspace{-10pt}
\end{table}

\begin{theorem}[Optimal Block Size]
Given a feature map size, encryption parameters, and CNN architecture, for \frag with block size $S$ and fully packed ciphertexts, the total rotation complexity per convolutional layer is minimized when $\frac{K^{2}}{S^{2}} = \frac{\alpha \cdot S^{2} \cdot N_{\text{out}}}{N_{\text{in}}}$, which defines the optimal block size:
\begin{equation}
\label{eq:optimal-S}
\footnotesize
\setlength\belowdisplayskip{1pt}
S^*=\lceil (\frac{K^2 N_{\text{in}}}{\alpha N_{\text{out}}} )^{1/4}  \rceil
    %\frac{K^{2}}{S^{2}} = \frac{N_c \cdot S^{2} \cdot N_{\text{out}}}{N_{\text{in}}},
\end{equation}
%which defines the optimal block size $S^*=\left(\frac{K^2 N_{\text{in}}}{N_c N_{\text{out}}}\right)^{1/4}$.
\end{theorem}

%\begin{proof}
\noindent\textit{\textbf{Proof.}}
The inner-rotation term decreases with $1/S^{2}$, as each ciphertext packs only $1/S^{2}$ of same-channel spatial correlated pixels, while the outer-rotation term increases proportionally with $S^{2}\alpha$ due to inter-channel dependency. Based on Equation~\ref{eq:total-rot}, the total rotation cost is the sum of a term decreasing in $S^{2}$ and a term increasing in $S^{2}$. The minimum occurs when these two terms are equal by the Cauchy--Schwarz inequality, yielding the condition above. Solving it gives the optimal $S$. \textit{This establishes the theoretical optimum of $S$ for minimal rotation complexity of \frag across any CNN model, dataset, and encryption setting-without runtime profiling, and aligns with the empirical results in Table~\ref{tbl:optimalS}.}
%\end{proof}

%\begin{bluebox}
\noindent\textbf{Analytical Complexity Comparison vs. SOTAs.}
% Table~\ref{matrixcompare} compares the amortized rotation complexity (per sample) of our optimal \frag configuration ($S = S^*$) against prior packing schemes-CHET, HElayers, and the recent SOTA Orion and Hyena+-using their published analytical models. For single-image inference ($BS{=}1$), our optimal encoding achieves the \textbf{lowest asymptotic complexity}, growing only linearly with the kernel size $K$, whereas all prior schemes scale with $K^2$. 
% Under batch inference, our complexity further decreases as $BS$ increases, benefiting from more effective amortization across packed samples. Overall, \frag \textbf{theoretically outperforms all existing schemes} in both single-sample and batched settings, with advantages that become even more pronounced at large batch sizes. These analytical findings are corroborated by our measured results in Fig.~\ref{fig:lat-mem-with-batch-new} and Table~\ref{tbl:detail-count-compare}.
Table~{\ref{matrixcompare}} compares the amortized rotation complexity (\textit{per sample}) of our optimal \frag configuration ($S = S^*$) with prior packing schemes-CHET, HElayers, Orion, Batchwise+ and Hyena+, using their published analytical models. For single-image inference ($BS=1$), our optimal encoding achieves the \textbf{lowest asymptotic complexity}, growing only linearly with the kernel size $K$, whereas prior schemes scale at least quadratically ($K^2$), even quartically ($K^4$) for Hyena+. This extra $K^2$ overhead arises because Hyena+ re-rotates the input ciphertext on the fly to sequentially generate $K^2$ output pixels, whereas other packing schemes reuse multiple pre-rotated ciphertext copies stored in memory to enable parallel output-pixel generation. For Batchwise+, each ciphertext packs a single pixel across multiple channels, resulting in a sparse format (the slot number is often much larger than the batch size, e.g. $2^{15} >$ 1024) and increasing the ciphertext count by $\lceil \nicefrac{M^2}{BS} \rceil\times$ compared to other methods. Moreover, its rotation cost is suboptimal due to increased outer rotations from denser channel packing within each ciphertext. Overall, Batchwise+ performs much worse than Hyena+, since $M^2 \gg K^2, K^4$ (e.g., $M=32$, $K=3$). This observation is well consistent with the Hyena paper\mbox{\cite{singh2024hyena}} and our experimental results in Fig.~{\ref{fig:lat-mem-with-batch-new}}.
Under batch inference, the rotation complexity of HELayers, Batchwise+, Hyena+, and our \frag further decreases as $BS$ increases, benefiting from improved amortization across packed samples. 
%\bbnote{Q12-F}
Nonetheless, \frag \textbf{theoretically outperforms all existing schemes} in both single-sample and batched settings, with advantages that become even more pronounced at large batch sizes. 
These analytical findings are corroborated by our measured results in Fig.~{\ref{fig:lat-mem-with-batch-new}} and Table~{\ref{tbl:detail-count-compare}}.

\subsection{Arch-aware Ct Compression}
\label{sec:channel}

While \frag minimizes rotations by tuning $S$, it assumes ciphertexts remain densely packed-an assumption often violated in modern CNNs. Channel-reduction layers (e.g., $1{\times}1$ convolutions in MobileNet~\cite{howard2017mobilenets}, SqueezeNet~\cite{iandola2016squeezenet}, and ResNet~\cite{he2016deep}) shrink intermediate featuremap channels, leaving many slots in HE-packed ciphertexts empty. Because CKKS applies SIMD operations across all slots, these sparse ciphertexts waste computation and require more ciphertexts in later layers. In a typical bottleneck, reducing channels from $N_{in}$ to $N_{DS}$ produces underutilized ciphertexts whenever $N_{DS}<\alpha$, where $\alpha$ is the packing capacity per ciphertext. Passing such ciphertexts into the expansion layer activates only a fraction of the SIMD units, forcing additional ciphertexts and rotations to cover $N_{out}$ output channels. As these low-utilization ciphertexts propagate through subsequent layers, both computation and memory costs grow unnecessarily.

To address this challenge, we introduce \multi (AAC), an architecture-aware ciphertext compression mechanism that restores slot density whenever channel-reduction layers would otherwise create sparsity. AAC reshapes the ciphertext so that each ciphertext entering a convolution contains as many valid channels as possible, maximizing utilization and reducing the number of ciphertexts required downstream. Crucially, AAC operates without altering the packing format and adapts automatically to all intermediate shapes, preserving the rotation bounds established by \frag and sustaining high-throughput HE execution under aggressive channel-reduction patterns.

%To address this challenge, we introduce \multi, an architecture-aware ciphertext compression (AAC) mechanism that restores slot density whenever channel-reduction layers would otherwise create sparsity. AAC repacks intermediate ciphertexts so that each ciphertext carries as many valid channels as possible, thereby maximizing SIMD utilization and reducing the number of ciphertexts required in subsequent layers. Importantly, AAC performs this reshaping without modifying the underlying \frag packing format and adapts automatically to all intermediate tensor shapes. As a result, it preserves the rotation-efficiency guarantees of \frag while sustaining high throughput even under aggressive bottleneck and channel-reduction patterns.

\begin{figure}[t!]
\centering
\includegraphics[width=0.8\columnwidth]{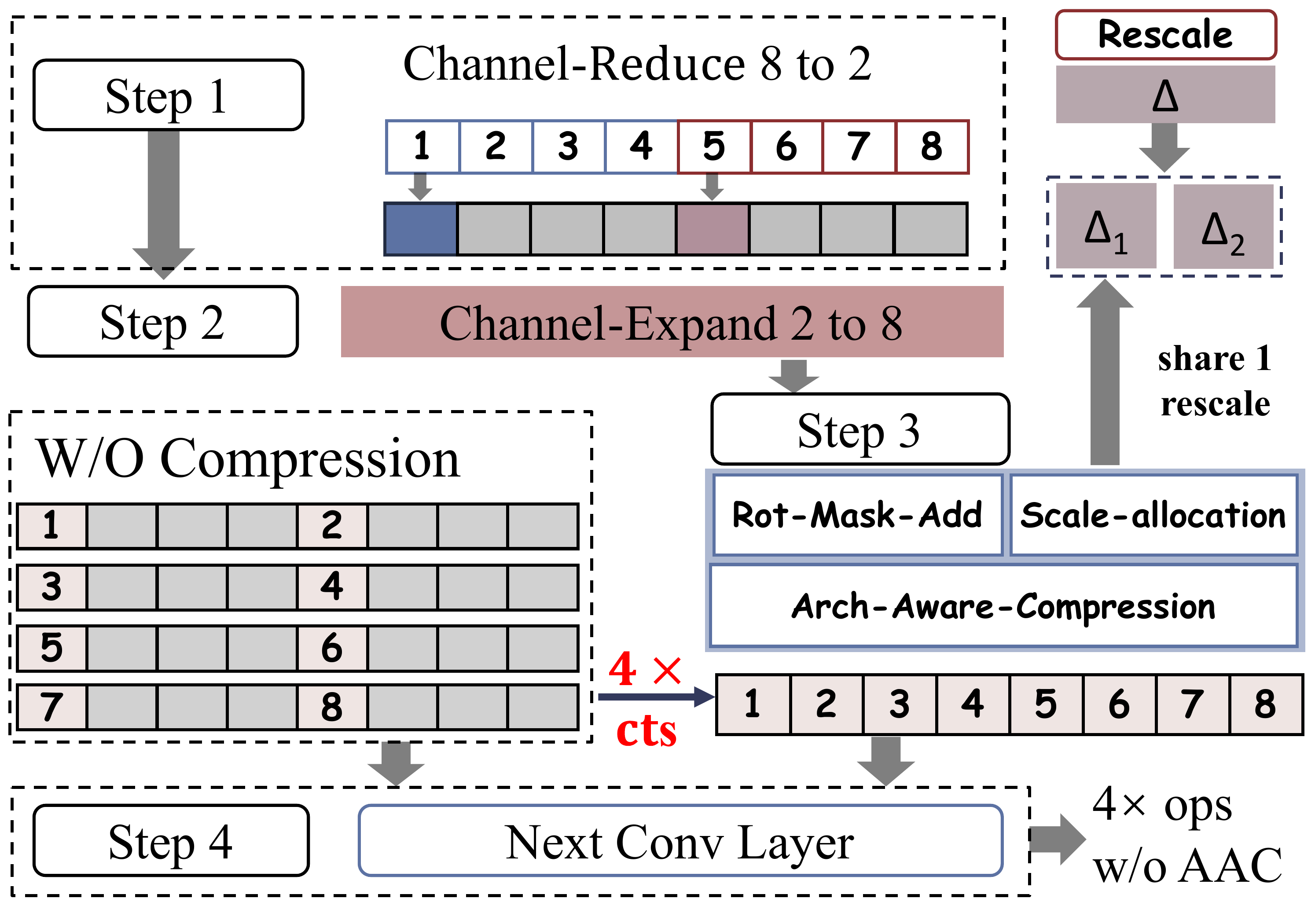}
% \vskip -0.10in
\small
\caption{Architecture-aware ciphertext compression. The example shows a bottleneck block ($8\!\to\!2\!\to\!8$ channels), where AAC consolidates the reduced 2-channel ciphertext into a fully packed form before the expansion layer.}
\label{fig:compression}
%\vskip -0.15in
\end{figure}

AAC performs a lightweight merge that compacts the valid channels into a fully packed ciphertext before the next convolution. In the bottleneck example of Figure~\ref{fig:compression} (step 1 to step 3, $8\!\to\!2\!\to\!8$ channels), the intermediate ciphertext after the $1\times1$ reduction contains only $2$ active channels and thus exhibits low slot utilization. AAC applies a small masked rotate-and-add sequence to consolidate these channels into a dense ciphertext, enabling the subsequent expansion layer to generate all $8$ output channels using only a single ciphertext. Without AAC, the same layer would require $4$ sparse ciphertexts, each with only $25\%$ utilization, incurring roughly
$4\times$ more HE computation in the next layer (Figure~\ref{fig:compression} step 4). While prior systems such as Fhelipe~\cite{krastev2024tensor}, Pantheon~\cite{pantheon2023pantheon}, and Coeus~\cite{ahmad2021coeus} use \texttt{rot-mask-add} patterns for data replication or communication efficiency, AAC repurposes this pattern specifically to preserve slot density across layer-wise dimension changes, maintaining effective SIMD parallelism throughout the HE-CNN pipeline.

Notably, AAC introduces no additional multiplicative depth, even though it applies a plaintext mask. The key observation is that the mask does \emph{not} require a higher plaintext scale than the convolution weights. In standard CKKS practice, each $PMult$ uses a uniform scale $\Delta$ (e.g., 40 bits) to meet precision requirements~\cite{kim2021hear,ran2022cryptogcn,peng2023lingcn}, regardless of the raw bit-width of the model parameters. Consequently, the two consecutive multiplications with the convolution weights and then by the AAC mask can share the same scale. As illustrated in \figureautorefname~\ref{fig:compression} step 3, the weight AAC mask is binary (0/1) and can be encoded at scale $\Delta_1$ and $\Delta_2$ separately so that $\Delta_1\cdot \Delta_2=\Delta$ without inflating precision or noise. Thus, both multiplications can be followed by a \emph{single} rescale applied only after the second multiplication. This avoids the extra multiplicative level that a naive \texttt{rot-mask-add} strategy would incur.

\section{Evaluation Methodology}
%\vspace{-0.1in}
\label{sec:eva-method}
% addressed \textcolor{red}{Wen: Please report the memory utilization, FLOPs in addition to latency, more system level stuff}\\
% \textcolor{blue}{
% Our encoding + other methods - e.g. Group Convolution\\
% Ob1: Level-higher ciphertext has higher latency\\
% Ob2: Rotation is more significant than other operations\\
% }

\begin{table}
% \vspace{-10pt}
\normalsize
\caption{The evaluated models and encryption parameters.}
% \vspace{-10pt}
\label{tbl:setting}
% \vskip 0.15in
\centering
\resizebox{0.9\columnwidth}{!}{
\begin{tabular}{c|cccc|c|c|c}
\hline
\multirow{2}{*}{Model} & \multicolumn{3}{c|}{\# Layers} & \multirow{2}{*}{N} & \multirow{2}{*}{Accuracy ($\%$)} & \multirow{2}{*}{Dataset} & \multirow{2}{*}{Kernel Size} \\
                       & Conv & FC & Act & \multicolumn{1}{|c|}{} &  &  \\ 
        \hline
LeNet & 2    & 2  & 3   & \multicolumn{1}{|c|}{$2^{15}$}  & 98.95 & MNIST    & 5x5\\
VGG5       & 4    & 1  & 4   & \multicolumn{1}{|c|}{$2^{15}$}  & 86.32 & CIFAR-10 & 3x3\\
SqueezeNet & 10   & 10 & 10  & \multicolumn{1}{|c|}{$2^{16}$}  & 81.5  & CIFAR-10 & 3x3 \& 1x1\\l
ResNet18   & 17   & 1  & 17  & \multicolumn{1}{|c|}{$2^{16}$}  & 66.8  & ImageNet & 3x3 \& 1x1 \\
MobileNet & 55  & 1 & 55  & \multicolumn{1}{|c|}{$2^{16}$}  & 72.0  & ImageNet & 3x3 \& 1x1\\
\hline
% \multirow{2}{*}{Model} & \multicolumn{4}{c|}{Encryption Parameters} & Mult  & Security  \\
%         & N & Q & P & q & Level &   Level   \\ 
% \hline
% LeNet  & 8192  &  218  & 26 & 36 & 7  & 128 bits\\
% VGG5       & 16384 &  438  & 29 & 32 & 14 & 128 bits\\
% ResNet18   & 65536 &  1501 & 46 & 51 & 58 & 153 bits \\
% \hline
\end{tabular}
}
\vspace{-0.1in}
\end{table}

\begin{figure*}[t]
%\begin{tcolorbox}[colback=blue!10, colframe=blue!30, sharp corners]
    \centering
    \includegraphics[width=.9\linewidth]{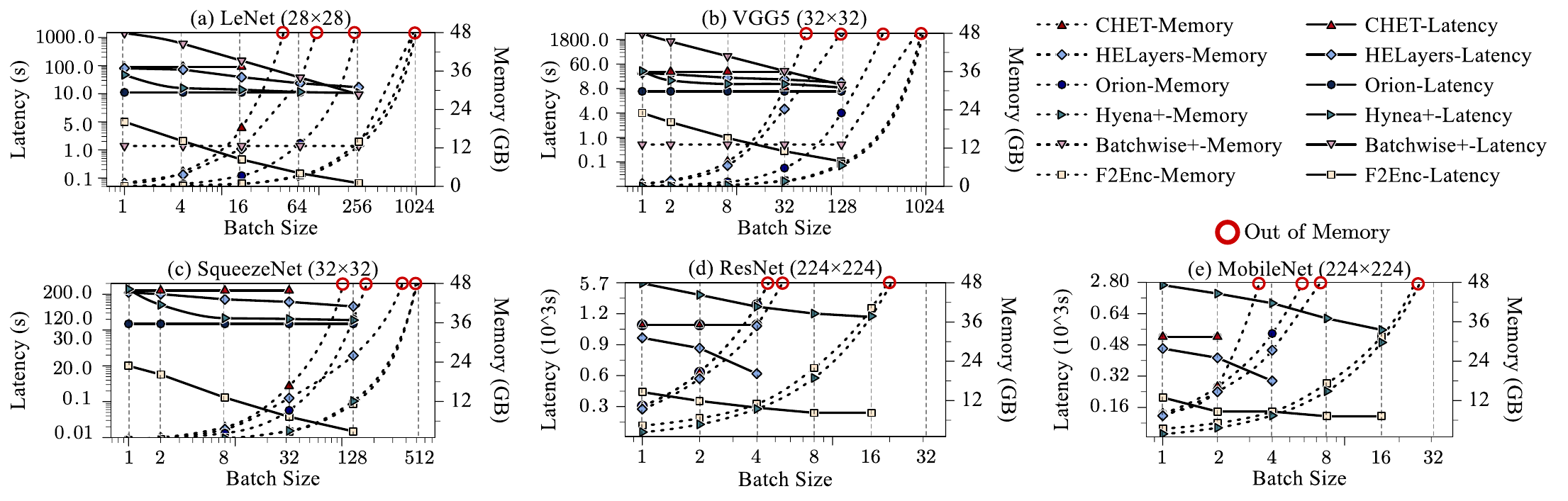}
    %\vspace{-5pt}
   %\caption{End-to-end latency and memory comparison of \ourmodel v.s. prior SOTAs across varying input sizes, batch sizes, and model scales on GPU.}
   \caption{Latency/memory comparison between \ourmodel and SOTAs across various input/batch sizes, and model scales on GPU. Note that Batchwise+'s results (in blue color) are available for LeNet (a) and VGG5 (b), but not reported for SqueezeNet (c), ResNet (d), and MobileNet (e) since it exhausts GPU memory under all batch size settings.}
    \label{fig:lat-mem-with-batch-new}
    %\vspace{-5pt}
%\end{tcolorbox}
\vspace{-15pt}
\end{figure*}
% \vspace{-20pt}

% \textbf{Configuration.} 
\noindent\textbf{Environment:} We conduct our experiments on a machine equipped with an AMD Threadripper 3975WX@3.5 GHz CPU, 256GB 8-Channel RAM running at 3200MT/s, and an RTX A6000 GPU with 48GB RAM. 
% \yx{CPU DRAM should has the number of channels and frequency to estimate the bandwidth it can provide}

% \noindent\textbf{Implementation:}
% Our system integrates a set of architecture-aware optimizations into an end-to-end encrypted CNN inference pipeline. During initialization, a \textit{conv-aware fragment encoding} module determines the optimal data layout based on the input and model structure, improving ciphertext packing efficiency and reducing rotation overhead. At runtime, \textit{arch-aware ciphertext-compression} optimizes the computation circuit by aligning with model topology and reducing ciphertext redundancy, complemented by leveraging weight-quantized model to offset the additional level consumption. Together, these components form a unified workflow that accelerates encrypted CNN inference. 

% \vspace{1mm}
\noindent\textbf{Implementation:}
We evaluate \ourmodel on the GPU backend using Liberate-FHE~\cite{Liberate_FHE}, an HE framework optimized for GPU execution. For deeper CNNs that require ciphertext refresh, we adopt the GPU-optimized bootstrapping from NEXUS~\cite{zhang2024secure}, where each bootstrapping consumes 14 ciphertext levels.

ReLU layers are replaced with the standard polynomial approximation $ax^{2} + bx + c$~\cite{kim2021hear,pmlr-v202-ran23b}. For fully connected layers, we use diagonal-matrix multiplication with BSGS optimization following prior work~\cite{ebel2025orion}. Notably, our \frag eliminates the need for the post-processing required in earlier schemes~\cite{kim2022secure} when handling stride $\ge2$ convolutions or average pooling: stride is implemented simply by discarding the unused cts. This avoids introducing vacant slots and reduces unnecessary HE computation. Moreover, if a subsequent convolution layer requires a different optimal block size $S$, we adapt the pre- or post-processing (\textit{rot–mask–add}) operations to adjust block size, following techniques similar to HEAR~\cite{kim2022secure}. This adjustment ensures that each layer is executed under its optimal rotation setting.

\noindent\textbf{Baselines:} We compare our method against four representative encrypted inference baselines:
\textbf{CHET}~\cite{dathathri2019chet}, \textbf{HELayers}~\cite{aharoni2023helayers}, 
and the recently proposed \textbf{Batchiwse+}~\cite{singh2024hyena}, \textbf{Hyena+}~\cite{singh2024hyena} and \textbf{Orion}~\cite{ebel2025orion}. Specifically, CHET introduces Toeplitz-based HE convolutions to exploit SIMD parallelism. 
HELayers, Batchwise+ and Hyena+ reduce rotation cost through block tiling and multi-image packing.
Orion further improves slot utilization and rotation efficiency using 
multi-channel packing with BSGS optimizations. Notably, Batchwise+ is also a variant when $S=M$ with our \frag.

\noindent\textbf{Models and Datasets:} We follow prior work~\cite{ebel2025orion,lee2022low,dathathri2019chet,aharoni2023helayers} and evaluate four representative models with their standard datasets: LeNet~\cite{dathathri2019chet,lecun2002gradient} on MNIST~\cite{lecun2010mnist}, VGG5~\cite{simonyan2014very} and SqueezeNet~\cite{iandola2016squeezenet} on CIFAR10~\cite{krizhevsky2010convolutional}, and ResNet18~\cite{lee2022low} and MobileNet~\cite{howard2017mobilenets} on ImageNet~\cite{deng2009imagenet}. Table~\ref{tbl:setting} summarizes the model structures and accuracies. Due to their higher computational depth, SqueezeNet, ResNet18, and MobileNet require bootstrapping, whereas LeNet and VGG5 do not require this step. To ensure a fair comparison, we adopt the bootstrapping placement of Orion~\cite{ebel2025orion} and the implementation of~\cite{zhang2024secure}; Our end-to-end timing measurements fully incorporate all bootstrapping overhead. 
% Finally, \textit{\ourmodel is evaluated on ImageNet-scale encrypted inference (ResNet/MobileNet), representing among the largest end-to-end HE workloads demonstrated to date and enabling a realistic assessment of encrypted deep learning at scale.}
Finally, we evaluate \textit{\ourmodel on ImageNet-scale encrypted inference (ResNet and MobileNet), representing one of the largest end-to-end homomorphic encryption workloads demonstrated to date and enabling a realistic assessment of secure deep learning at scale.}

\noindent\textbf{Encryption Parameters.} Table~\ref{tbl:setting} provides the encryption parameters $N$ used for evaluation each model adopted in RNS-CKKS. In our experiments, we use a fixed scale factor $\Delta = 2^{40}$ (40 bits) for ciphertext encoding to maintain numerical precision, and select the appropriate ciphertext modulus $Q$ to guarantee a security level $\lambda \geq 128$ bits for all evaluated models, sufficient to withstand the known attacks in~\cite{albrecht2015concrete}.

\section{Evaluation}
\label{Eval}

%\subsection{End-to-End Results}
% \textcolor{red}{Wen: You need to re-write the whole thing, the number is even not consistent with the table 3, there are so many redundant and unprofessional description in the evaluation.}

% Table~\ref{tbl:eval-sota} presents the results of our method compared to previous SOTA methods. We use mainly two metrics: the number of processed image per second (img/sec) and memory footprint per image (mem/img) to compare end-to-end performance. 
%To ensure a fair comparison, we maximum batch sizes that SOTA methods can support without memory swapping, as these methods require substantially more memory for intermediate results.
%The maximum batch sizes that SOTA methods can support on LeNet, VGG5, ResNet18, and SqueezeNet are 256, 128, 8, and 128, respectively. Results with the minimum batch sizes are also included.
%As shown in Figure~\ref{tbl:eval-sota}, our method outperforms the selected baselines in both img/sec and mem/img (GB) across all benchmarks. 
%Specifically, it speeds up img/sec by $1846.33\times$, $12.25\times$, and $8.13\times$, and reduce mem/img by $\times$, $\times$, and $\times$, compared to CHET, HELayers, and Orion, respectively. The performance improvement is most significant for rotation reduction, better slot utilization and smaller encryption parameter. 

\begin{table*}[t]
%\tiny
\centering
\caption{Breakdown of HE operations. Percentages for \ourmodel indicate reductions relative to the baseline \emph{HElayers}.}
\label{tbl:detail-count-compare}
\resizebox{0.85\textwidth}{!}{
\begin{tabular}{c|c|c|c|c|c|c|c|c}
\hline
Model &
Batch &
HE-CNN Schemes &
Rot &
\begin{tabular}{@{}c@{}}Key-\\switching\end{tabular} &
\begin{tabular}{@{}c@{}}NTT \\ \& INTT \end{tabular} & 
\begin{tabular}{@{}c@{}}ct\\count\end{tabular} &
Mult &
\begin{tabular}{@{}c@{}}Slot Util.\\(\%)\end{tabular} \\ 
\hline
\multirow{8}{*}{SqueezeNet}  & \multirow{4}{*}{batch=1}   & HElayers & 1065 & 1225 & 664K  & 16  & 331K  & 12.5  \\ 
                             &                            & Orion    & 674  & 754  & 287K  & 8   & 143K  & 25    \\ 
                             &                            & \ourmodel-F   & 477  & 517  & 176K  & 4   & 877K   & 50    \\
& & \cellcolor[gray]{0.9} \textbf{\ourmodel}  
& \cellcolor[gray]{0.9} 224 ($\downarrow$79\%) 
& \cellcolor[gray]{0.9} 244 ($\downarrow$80\%) 
& \cellcolor[gray]{0.9} 72K ($\downarrow$89\%) 
& \cellcolor[gray]{0.9} 2 ($\downarrow$88\%)  
& \cellcolor[gray]{0.9} 36K ($\downarrow$89\%) 
& \cellcolor[gray]{0.9} 100 \\
                             \cline{2-9} 
                             & \multirow{4}{*}{batch=Max} & HElayers & 834  & 994  & 664K  & 16  & 331K  & 12.5  \\ %\cline{3-9} 
                             &                            & Orion    & 674  & 754  & 287K  & 8   & 143K  & 25    \\ %\cline{3-9} 
                             &                            & \ourmodel-F   & 92   & 132  & 175K  & 4   & 877K   & 50    \\ %\cline{3-9} 
& & \cellcolor[gray]{0.9} \textbf{\ourmodel}     
& \cellcolor[gray]{0.9} 54 ($\downarrow$94\%)   
& \cellcolor[gray]{0.9} 74 ($\downarrow$93\%)   
& \cellcolor[gray]{0.9} 71K ($\downarrow$89\%)  
& \cellcolor[gray]{0.9} 2 ($\downarrow$88\%)    
& \cellcolor[gray]{0.9} 35K ($\downarrow$89\%)  
& \cellcolor[gray]{0.9} 100 \\
                             \hline
\multirow{8}{*}{ResNet18}    & \multirow{4}{*}{batch=1}   & HElayers & 2398 & 2758 & 1409K & 36  & 703K  & 31.25 \\ %\cline{3-9} 
                             &                            & Orion    & 1894 & 2134 & 878K  & 24  & 438K  & 50    \\ %\cline{3-9} 
                             &                            & \ourmodel-F   & 1040 & 1160 & 439K  & 12  & 218K  & 75    \\ %\cline{3-9} 
& &  \cellcolor[gray]{0.9} \textbf{\ourmodel}
&  \cellcolor[gray]{0.9} 788 ($\downarrow$67\%)
&  \cellcolor[gray]{0.9} 868 ($\downarrow$69\%)
&  \cellcolor[gray]{0.9} 159K ($\downarrow$89\%)
&  \cellcolor[gray]{0.9} 8 ($\downarrow$78\%)
&  \cellcolor[gray]{0.9} 79K ($\downarrow$89\%)
&  \cellcolor[gray]{0.9} 100 \\
                             \cline{2-9} 
                             & \multirow{4}{*}{batch=Max} & HElayers & 2046 & 2406 & 1408K & 36  & 703K  & 31.25 \\ %\cline{3-9} 
                             &                            & Orion    & 1894 & 2134 & 878K  & 24  & 438K  & 50    \\ %\cline{3-9} 
                             &                            & \ourmodel-F   & 502  & 622  & 439K  & 12  & 2189K  & 75    \\ %\cline{3-9} 
& &  \cellcolor[gray]{0.9} \textbf{\ourmodel}
&  \cellcolor[gray]{0.9} 242 ($\downarrow$88\%)
&  \cellcolor[gray]{0.9} 322 ($\downarrow$87\%)
&  \cellcolor[gray]{0.9} 158K ($\downarrow$89\%)
&  \cellcolor[gray]{0.9} 8 ($\downarrow$78\%)
&  \cellcolor[gray]{0.9} 78K ($\downarrow$89\%)
&  \cellcolor[gray]{0.9} 100 \\ 
                             \hline
\multirow{8}{*}{MobileNet} & \multirow{4}{*}{batch=1}   & HElayers & 7044 & 9604 & 8161K & 256 & 4076K & 7.25  h\\ %\cline{3-9} 
                             &                            & Orion    & 5391 & 6671 & 4082K & 128 & 2038K & 12.5  \\ %\cline{3-9} 
                             &                            & \ourmodel-F   & 2696 & 3336 & 2041K & 64  & 1019K & 25    \\ %\cline{3-9} 
& &  \cellcolor[gray]{0.9} \textbf{\ourmodel}
&  \cellcolor[gray]{0.9} 786 ($\downarrow$89\%)
&  \cellcolor[gray]{0.9} 946 ($\downarrow$90\%)
&  \cellcolor[gray]{0.9} 510K ($\downarrow$94\%)
&  \cellcolor[gray]{0.9} 16 ($\downarrow$94\%)
&  \cellcolor[gray]{0.9} 255K ($\downarrow$94\%)
&  \cellcolor[gray]{0.9} 100 \\
                             \cline{2-9} 
                             & \multirow{4}{*}{batch=Max} & HElayers & 5892 & 8452 & 8160K & 256 & 4075K & 7.25  \\ %\cline{3-9} 
                             &                            & Orion    & 5391 & 6671 & 4082K & 128 & 2038K & 12.5  \\ %\cline{3-9} 
                             &                            & \ourmodel-F   & 1894 & 2534 & 2040K & 64  & 1019K & 25    \\ %\cline{3-9} 
& &  \cellcolor[gray]{0.9} \textbf{\ourmodel}
&  \cellcolor[gray]{0.9} 472 ($\downarrow$92\%)
&  \cellcolor[gray]{0.9} 632 ($\downarrow$93\%)
&  \cellcolor[gray]{0.9} 510K ($\downarrow$94\%)
&  \cellcolor[gray]{0.9} 16 ($\downarrow$94\%)
&  \cellcolor[gray]{0.9} 255K ($\downarrow$94\%)
&  \cellcolor[gray]{0.9} 100 \\
                             \hline
\end{tabular}
}
\vspace{-15pt}
\end{table*}

\subsection{End-to-End Performance Evaluation} 
\noindent\textbf{Performance evaluation on GPU.} We first evaluate the end-to-end performance of \ourmodel on a GPU platform (Section~\ref{sec:eva-method}).
Figure~\ref{fig:lat-mem-with-batch-new} reports the amortized latency (sec/image) and the total memory footprint under different batch sizes across various CNN models. To isolate the contribution of each component, we define \textbf{\ourmodel-F} as the variant using only fragment encoding, while \textbf{\ourmodel} applies both proposed techniques jointly.

Overall, \ourmodel consistently outperforms all prior encrypted-inference systems across the evaluated CNN architectures. Averaged across all benchmarks, \ourmodel achieves \textbf{$748.63\times$}, \textbf{$213.02\times$}, %\mynew{\textbf{$633.08\times$}}, 
and \textbf{$109.96\times$} latency speedups on average over \texttt{CHET}, \texttt{HELayers}  %\mynew{\texttt{Hyena+}}, 
and \texttt{Orion}, respectively. % \textcolor{red}{minor question: since we removed the bar graph with exact number annotations, the figure now does not have exact numbers on them. Would the readers concern about how to read the figures to get the data we present here?} 
In addition, \ourmodel significantly extends the maximum batch size that can be accommodated under the same GPU memory capacity, up to $512$, $512$, $256$, $16$, and $16$ for LeNet, VGG5, SqueezeNet, ResNet18, and MobileNet, respectively. 

%Note Hyena+ achieves the similar memory efficiency as \ourmodel, \mynew{This is close to Hyena+, whose high memory efficiency stems from re-rotating input ciphertexts on the fly to generate each output pixel sequentially, allowing operation within very small on-chip caches instead of relying on large memory to store multiple pre-rotated copies for parallel pixel generation.}

%This is comparable with that of Hyena+, whose high memory efficiency is guaranteed via re-rotating an input ciphertext on the fly to produce each output pixel sequentially, to fit into very small on-chip caches instead of using large memory to store multiple pre-rotated copies for generating output pixel in parallel.  

%instead of using large accessing pre-roated ciphertext copies in memory, to 

%\mynew{Hyena+ achieves comparable memory efficiency as \ourmodel with respect to batch size . 
%This is because
%Hyena re-rotate the input ciphertexts to avoid storing the inner-rotated-copies, while \ourmodel principally reduce the inner-rotation complexity to achieve less inner-rotated copies.
%} %\bbnote{Q1-F}
%\mynew{Except for Hyena+, it also  it is designed for low memory footprint with tradeoff the computation cost.}        

For small CNNs such as LeNet and VGG5, our framework delivers the largest performance gains. These networks operate on relatively small feature maps (32$\times$32), allowing \ourmodel to pack substantially larger batches into a single ciphertext and thus exploit much higher SIMD parallelism than settings limited to small batch counts (e.g., 1, 4, 16). Consequently, these scenarios highlight the performance divergence between \ourmodel and schemes like \texttt{Hyena+} and \texttt{Batchwise+}, as illustrated in Fig.~\ref{fig:lat-mem-with-batch-new}.
%These settings also make the differences relative to \texttt{Hyena+} and \texttt{Batchwise+} particularly clear in Fig.~\ref{fig:lat-mem-with-batch-new}. 
As a result, for \texttt{LeNet}, it achieves the highest speedup, realizing up to \textbf{$1846\times$}, \textbf{$1625.5\times$}, and $228.8\times$ acceleration (0.0493s) compared to 91.13s, 80.14s, %\mynew{10.73s}, 
and 11.19s reported by \texttt{CHET}, \texttt{HELayers}, %\mynew{Hyena+}, 
and \texttt{Orion}, respectively; it is also \textbf{$264.32\times$} faster than \texttt{Hyena+} (10.73s at batch size 256) and up to \textbf{$231.7\times$} faster than \texttt{Batchwise+} across the evaluated batch sizes. VGG5 shows similar trends: at batch 128, \ourmodel improves throughput by \textbf{$626.25\times$}, \textbf{$171.25\times$}, %\mynew{\textbf{$264.32\times$}}, 
and \textbf{$86.92\times$} over the three baselines. The mechanisms behind these gaps differ. \texttt{Hyena+} attains strong memory efficiency by re-rotating inputs on the fly, but its sequential output generation and neighbor-preserving packing still incur inefficient inner rotations, leading to much higher latency, consistent with the $O(K^4)$ vs. $O(K)$ gap in Table~\ref{matrixcompare}. \texttt{Batchwise+}, in contrast, removes inner rotations but decomposes dense ciphertexts into pixel-wise ciphertexts, which inflates outer rotations, multiplications/additions, and the ciphertext count by $\lceil \nicefrac{M^2}{BS} \rceil\times$.

For the larger SqueezeNet model, although the input feature-map resolution remains 32$\times$32, the memory footprint per ciphertext increases due to larger encryption parameters (see \tableautorefname~\ref{tbl:setting}). This limits the maximum feasible batch size from $1024$ to $512$, thereby constraining the attainable batch-level parallelism. Moreover, the prevalence of $1{\times}1$ convolutions substantially reduces inner-rotation overhead; hence, the gains here primarily stem from eliminating ciphertext redundancy via AAC rather than rotation reduction. Even in this low-rotation regime, \ourmodel still yields substantial improvements, achieving \textbf{$215.15\times$} speedup over \texttt{CHET} and \textbf{$59.10\times$} over \texttt{Orion}. This demonstrates that our optimizations generalize beyond rotation-heavy workloads and remain effective for architectures dominated by efficient $1{\times}1$ convolutions. As Table~\ref{tbl:detail-count-compare} shows, with AAC applied, ciphertext slot utilization increases from $50\%$ to nearly $100\%$, confirming the efficiency of our packing strategy. This also explains why \texttt{Batchwise+} can be evaluated only for the smaller LeNet/VGG5 settings in Fig.~\ref{fig:lat-mem-with-batch-new}: once deeper HE circuits require a larger polynomial degree ($N = 2^{16}$), the per-ciphertext memory and computation overhead rises sharply, and its $\lceil \nicefrac{M^2}{BS} \rceil\times$ ciphertext amplification quickly becomes prohibitive.

For high-resolution ImageNet on larger models like ResNet18 and MobileNet, prior HE inference systems are fundamentally constrained by ciphertext capacity: a single-channel feature map can nearly saturate the available HE slots. In this case, Orion's multi-channel optimization becomes inapplicable, causing its data layout to degenerate into the same format as \texttt{CHET}.
%Orion's  optimization no longer exists, translated to the same format as CHET. 
By decoupling spatial and channel dependencies, \ourmodel retains efficient multi-channel packing while reducing rotation and multiplication overhead. On ResNet18, this design yields \textbf{$8.93\times$}, \textbf{$3.47\times$}, and \textbf{$8.93\times$} improvement over three baselines (\texttt{CHET}, \texttt{HELayers}, and \texttt{Orion}) at batch size 4; even relative to the more memory-efficient \texttt{Hyena+}, \ourmodel remains \textbf{$4.85\times$} faster at batch size 16 (240s vs. 1168.45s). For MobileNet, \ourmodel not only accelerates computation but also supports up to \textbf{4$\times$ larger} batch sizes before reaching GPU memory limits, highlighting its scalability advantage. For
\texttt{HELayers}, the method partitions large feature maps into smaller blocks and incorporates batching to fill ciphertext slots, which can reduce overhead and improve latency as batch size increases. However, its limited ability to reuse wasted slots causes bottleneck convolution layers to constrain slot utilization to 31.25\% and 7.25\% in ResNet and MobileNet, respectively. Therefore, its overall throughput remains limited.  
%\mynew{For Hyena, its design does not reduce the inner-rotation complexity because neighboring pixels are still packed within the same ciphertext. Without increasing the batch size, its inference latency is even higher than that of HE Layers.}
\textbf{\ourmodel consistently achieves higher throughput as batch size increases, because conv-aware encoding and AAC dramatically reduce rotation cost and memory usage. This enables far larger parallelism and pushes the OOM boundary well beyond prior encrypted-inference systems.} Besides, for ResNet18 and MobileNet, bootstrapping is required during inference, introducing an additional $191\,\text{s}$ and $105\,\text{s}$ latency, respectively. Although this unavoidable overhead increases the total runtime, our method substantially reduces rotation cost across convolutional layers, enabling significant overall acceleration even under bootstrapping.

%Besides, for ResNet18 and MobileNet, bootstrapping is required for inference. it consume 191s and 105s, respectively. Although it added extra latency overhead, our method reduces the rotation overhead in extensive convolutional layers. 

Fig.~\ref{fig:lat-mem-with-batch-new} highlights a clear scalability gap among different schemes: \ourmodel consistently supports substantially larger batch sizes before encountering GPU memory limits. Across all networks, throughput-oriented baselines such as \texttt{CHET}, \texttt{HELayers}, and \texttt{Orion} exhaust memory at relatively small batch sizes, whereas \ourmodel continues to scale until the hardware limit is reached. For example, LeNet and VGG5 run out of memory at batch sizes of approximately 64–128 and 128 under Orion, while \ourmodel sustains batches of 256 and 512 respectively. SqueezeNet shows a similar trend, where baselines stop at 64–128 but \ourmodel remains stable at 128. The divergence is most pronounced on MobileNet: \texttt{CHET}, \texttt{HELayers}, and \texttt{Orion} fail at batch sizes as small as 4–8, while \ourmodel operates reliably at 16–32. Compared with these baselines, \texttt{Hyena+} further extends the feasible batch-size range by re-rotating ciphertexts on the fly, but this memory reduction comes at the cost of substantially higher latency; \texttt{Batchwise+} exhibits the opposite tradeoff, namely low inner-rotation cost but rapidly growing ciphertext count and poor scalability once model depth and polynomial degree increase. This enlarged OOM boundary directly stems from our reduced ciphertext count and significantly improved slot utilization, which together minimize per-image memory consumption and enable efficient parallelism under limited GPU memory. Consequently, \ourmodel maintains high throughput across batch sizes that prior encrypted-inference systems cannot support. 

To analyze the sources of performance gains, Table~\ref{tbl:detail-count-compare} reports amortized metrics per input sample including rotation counts, slot utilization, ciphertext count, plaintext multiplications, key-switching, and NTT/INTT operations. For clarity, we present results in two modes: single-image inference with batch size 1, and high-throughput inference with batch sizes maximized to the GPU memory limits of each benchmark. 
%\mynew{We exclude \texttt{CHET} and \texttt{Hyena} for this experiment. \texttt{CHET} exhibits operational trends and slot utilization patterns similar to \texttt{HElayers}. Furthermore, \texttt{Hyena} does not propose dedicated techniques for improving slot utilization following feature-reduction layers. As its threat model assumes an HE-MPC setting, \texttt{Hyena} relies on decryption and re-encryption to manage slot utilization, which is fundamentally distinct from our pure FHE-based design.}

%\texttt{CHET} suffers from extremely low slot utilization ($<15.6\%$) due to single-channel packing.
% \texttt{HELayers} partially alleviates rotation overhead using blocking, reaching a slot utilization of $31.25\%$.
% \texttt{Orion} improves utilization to $56.25\%$ via multi-channel multiplexing but is limited by channel-dependent packing, especially for large feature maps (e.g., ResNet18 with $224{\times}224$), where a single-channel tensor exceeds ciphertext capacity and prevents further rotation reduction.
While \texttt{HELayers} partially alleviates rotation overhead using spatial blocking-reaching a slot utilization of $31.25\%$, \texttt{Orion} increases this to $56.25\%$ via multi-channel multiplexing. However, \texttt{Orion}'s design remains constrained by channel-dependent packing; for large input feature maps (e.g., ResNet18 with $224 \times 224$), a single-channel tensor can exceed the ciphertext capacity, which restricts further rotation reduction. In contrast, by fully decoupling both spatial and channel dependencies, \ourmodel achieves the lowest rotation count and the highest slot utilization. Specifically, \ourmodel-F achieves a slot utilization of $62.5\%$, while \ourmodel-F+C reaches a perfect $100\%$ utilization and reduces average rotation operations by $88.55\%$. For \texttt{LeNet} and \texttt{VGG5}, \ourmodel-F alone reaches the optimal $100\%$ utilization, yielding a $94\%$ to $98\%$ reduction in rotation count compared to \texttt{Orion}. Furthermore, this reduction in rotation operations (\textbf{Rot}) directly propagates to lower-level FHE primitives. Compared to \texttt{HELayers}, \ourmodel reduces key-switching and NTT/iNTT transformations by up to \textbf{93\%--94\%}. Concurrently, high slot utilization translates to a \textbf{89\%--94\%} reduction in both ciphertext count and homomorphic multiplication operations. These results confirm that \ourmodel synergistically optimizes rotation overhead and ciphertext slot utilization, significantly minimizing the overall FHE workload exposed to the underlying hardware.

\vspace{1mm}
\noindent\textbf{System Efficiency Analysis.} 
We evaluate how our design gains benefit from improved system/hardware efficiency with two representative benchmarks: LeNet and SqueezeNet. To be specific, on LeNet, with batch size 32, Table~{\ref{tbl:hw_compare}} compares the baseline-Orion execution with {\ourmodel} best fragmented configuration (block = 4). The optimized setting reduces (i) NTT kernel invocations from 48,015 to 5,775 (88.0\% reduction), and GPU memory transfer size from 12,021 MB to 1,461 MB (87.8\% reduction), and the number of memory transfer operations from 6,464 to 832 (87.1\% reduction). Overall, the results show consistent optimization benefits across configurations, reducing kernel invocations and memory traffic. A similar trend appears on SqueezeNet. As shown in Table~{\ref{tbl:hw_compare_blocksize}}, increasing the block size from 2 to 4 reduces kernel invocations from 69,136 to 7,984, memory traffic from 15,167 MB to 1,790 MB, and memory transfer operations from 8,728 to 1,084. 
Overall, the experiments show that {\ourmodel}'s packing mechanism is an effective optimization strategy across both design-space exploration and end-to-end evaluation. 
%Accordingly, hardware accelerators should provision resources for the optimal operating region with adaptive block sizing, rather than sizing for worst-case behavior. The 544$\times$ gap in kernel calls between the worst and optimal configurations underscores the difference between practical and impractical deployment. For a detailed discussion regarding the implications for hardware architecture design, please refer to Sec.~{\ref{sec:indication_hw}}.
%\bbnote{Q2-B}

\begin{table}[tbp]
\centering
% \begin{tcolorbox}[
%     colback=blue!10,      % 浅蓝色背景
%     colframe=blue!10,     % 边框设为与背景同色
%     sharp corners,        % 论文通常使用直角
%     boxrule=0pt,          % 去掉边框线条
%     left=5pt, right=5pt, top=2pt, bottom=0pt, % 内部留白，防止文字贴边
%     halign=center         % 内部内容居中
% ]
    \caption{Comparison of system efficiency between Baseline-Orion and $Fenc^2$.}
    \label{tbl:hw_compare} % 在这里设置你的 label

    \resizebox{\columnwidth}{!}{%
        \begin{tabular}{c|c|c|c}
        \hline
        Config         & Kernel Calls & Memory Transfer Size (MB) & No. of Memory Transfers \\ \hline
        Baseline-Orion & 48,015       & 12,021                    & 6,464                   \\ \hline
        $Fenc^2$       & 5,775        & 1,461                     & 832                     \\ \hline
        Reduction      & 88.0\%       & 87.8\%                    & 87.1\%                  \\ \hline
        \end{tabular}%
    }
%\end{tcolorbox}
%\vspace{-10pt}
\end{table}

\begin{table}[tbp]
\centering
% \begin{tcolorbox}[
%     colback=blue!10,      % 浅蓝色背景
%     colframe=blue!10,     % 背景与边框颜色一致
%     sharp corners,        % 直角边框
%     boxrule=0pt,          % 去掉边框线
%     left=5pt, right=5pt, top=0pt, bottom=2pt, % 内部留白
%     halign=center         % 内部内容居中
% ]
    \caption{ Impact of block sizes on system efficiency for SqueezeNet.}
    \label{tbl:hw_compare_blocksize} % Label 占位符

    \resizebox{\columnwidth}{!}{%
        \begin{tabular}{c|c|c|c}
        \hline
        Block     & Kernel Calls & Memory Transfer Size (MB) & No. of Memory Transfers \\ \hline
        4         & 69,136       & 13,000                    & 8,728                   \\ \hline
        8         & 7,984        & 1,534                     & 1,084                   \\ \hline
        Reduction & 88.5\%       & 88.2\%                    & 87.6\%                  \\ \hline
        \end{tabular}%
    }
%\end{tcolorbox}
%\vspace{-15pt}
\end{table}

\vspace{1mm}
\noindent\textbf{Performance evaluation on CPU.}
Table~\ref{tbl:eval-sota-cpu-max} further demonstrates that \ourmodel generalizes effectively to CPU environments. Given the trends are similar to GPU results, we report here only the comparison with the strongest baseline-Orion, using the maximum supported batch size. Owing to the significantly larger memory capacity available on CPU platforms (e.g., 256\,GB), \ourmodel can support much larger batch sizes, leading to even greater throughput benefits compared to \texttt{Orion} (e.g., $226.06\times$ speedup on LeNet and $146.92\times$ on VGG5). While these results highlight the scalability of our approach beyond GPU platforms, \textbf{for consistency, all subsequent experiments are conducted on GPUs}, which are more commonly used for encrypted model inference workloads.

\begin{table}[b]
\scriptsize
\centering
\caption{CPU-side comparison with \texttt{Orion} at maximum supported batch size (256\,GB memory).}
\label{tbl:eval-sota-cpu-max}
\resizebox{.95\columnwidth}{!}{
\begin{tabular}{c|c|c|c|c|c}
\hline
Model &
  Method &
  Mem (GB) &
  Lat (s) &
  Speedup &
  Mem $\downarrow$ (\%) \\ \hline

\multirow{2}{*}{LeNet} 
    & Orion         & 0.21  & 40.87  & -        & -        \\ 
    & \textbf{\ourmodel} & 0.013 & 0.18   & 226.06   & 98.49\%  \\ \hline

\multirow{2}{*}{VGG5}
    & Orion         & 0.18  & 35.26  & -        & -        \\
    & \textbf{\ourmodel} & 0.012 & 0.24   & 146.92   & 73.53\%  \\ \hline

\multirow{2}{*}{SqueezeNet}
    & Orion         & 0.29  & 235.2  & -        & -        \\
    & \textbf{\ourmodel} & 0.016 & 3.98   & 59.10    & 60.62\%  \\ \hline

\multirow{2}{*}{ResNet18}
    & Orion         & 10.3  & 3930   & -        & -        \\
    & \textbf{\ourmodel} & 1.51  & 442    & 8.93     & 87.13\%  \\ \hline

\multirow{2}{*}{MobileNet}
    & Orion         & 8.11  & 3094   & -        & -        \\
    & \textbf{\ourmodel} & 1.21  & 328    & 9.43     & 85.08\% \\ \hline

\end{tabular}
}
\\[-2pt]
    \parbox{\textwidth}{\vspace{1ex}\hspace{5ex}\tiny\textit{Table notes}: Max batch sizes: 1024 (LeNet), 1024 (VGG5), 512 (SqueezeNet), 16 (ResNet18), 16 (MobileNet).}
\vspace{-5pt}
\end{table}

\begin{figure*}
\centering
\includegraphics[width=1\linewidth]{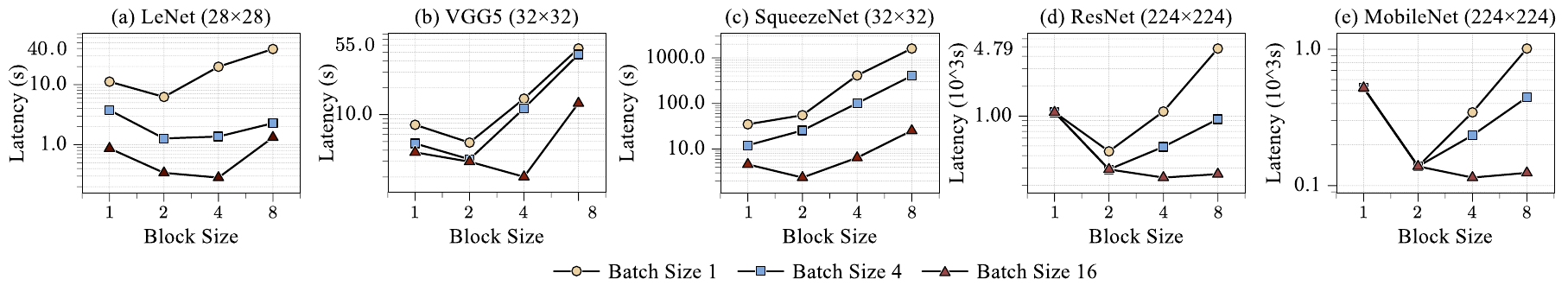}
%\vspace{-12pt}
\caption{Sensitivity analysis of the block size $S$ under different input sizes, model scales and batch sizes.}
\label{fig:lat-mem-with-block-size}
%\vspace{-15pt}
\end{figure*}

\subsection{Applicability of \ourmodel across Diverse CNN Configurations}
\label{sec:detailed-ab-study}

In this section, we investigate how \ourmodel performs in different CNN models considering various block size selections. Fig.~\ref{fig:lat-mem-with-block-size} compares the end-to-end latency under different block sizes across various CNN models and input resolutions. 

\noindent\textbf{Large batch size benefit significantly from larger block sizes.}
Across all five benchmarks, we consistently observe that the benefit of using a larger block size becomes more prominent as the batch size increases. For LeNet and VGG5, the latency gap between small blocks ($S{=}1,2$) and a larger block size ($S{=}4$) widens substantially when batch size grows from 1 to 16, with $S{=}4$ giving the lowest latency at higher batch sizes. A similar pattern appears in SqueezeNet, where larger $S$ provides minimal gains at batch size 1 but yields significant improvements at batch sizes 4 and 16. This trend is even clearer in ResNet18 and MobileNet: because their 224×224 feature maps incur higher inner-rotation cost, the reduction in rotation complexity brought by a larger block size only dominates when enough images are processed in parallel. Overall, larger batch sizes consistently amplify the advantage of larger block sizes, as the reduced inner-rotation overhead can be better amortized across more inputs.

\noindent\textbf{The benefit of larger block sizes grows with feature-map resolution.}
For a fixed batch size, models with larger feature maps (e.g., the
$224{\times}224$ feature maps in ResNet and MobileNet) experience substantially
greater latency reduction as the block size $S$ increases. This is because, for
large $M$, the computation is dominated by \emph{outer rotations} across the
feature map, while the \emph{inner-rotation} cost becomes relatively negligible. 
Thus, increasing $S$ decreases the number of blocks ($m^2 = (M/S)^2$), directly reducing the total rotation overhead. This trend is prominent in both ResNet and MobileNet, where selecting $S \in \{4, 8\}$ consistently yields lower latency than configurations with smaller block sizes. In contrast, models with smaller feature maps (e.g., LeNet and VGG5, operating on $28 \times 28$ or $32 \times 32$ inputs) exhibit more moderate improvements; their rotation overhead is inherently smaller and less sensitive to variations in block count.
% Thus, increasing $S$ reduces the number of blocks
% ($m^2 = (M/S)^2$), directly decreasing the total rotation cost. This trend is clearly observed in both ResNet and MobileNet: selecting
% $S \in \{4,8\}$ consistently yields lower latency than smaller block sizes.  
% In contrast, models with smaller feature maps (e.g., LeNet and VGG5 with
% $28{\times}28$ or $32{\times}32$ inputs) show more moderate improvement, since
% their rotation overhead is inherently smaller and less sensitive to the change
% in block count.

\noindent\textbf{Larger convolution kernels amplify the effect of block-size scaling.}
Models with larger convolution footprints require more inner-rotations, making
them more sensitive to fragment reduction. This trend is clearly reflected in
LeNet ($K{=}5$), where the latency drops more noticeably as $S$ increases
(e.g., between $S{=}1$ and $S{=}4$ at batch size $1$ and $16$), compared to VGG5
($K{=}3$), which exhibits a milder improvement across corresponding block sizes.
The larger kernel introduces more rotation paths per ciphertext, and thus
reducing the number of spatial fragments yields greater savings. Consequently,
block-size scaling provides stronger benefits in models with larger convolution
kernels.

\noindent\textbf{Optimal Block Size Consistently Obtainable Across Benchmarks.}
We observe that for each evaluation case there exists a block size $S$ that minimizes latency, and this optimal point shifts depending on feature-map size, kernel size, and batch size. It is consistent with our theoretical optimal block-size analysis (Section~\ref{sec:optimal}). For lightweight models with small feature maps (e.g., LeNet, VGG5), the improvement from increasing $S$ is marginal. In contrast, models with larger feature maps or larger kernels (e.g., ResNet18, MobileNet, LeNet) show a clear latency reduction as $S$ increases, especially under larger batch settings. Overall, larger block sizes become increasingly advantageous when the computation becomes heavier, and the minimum latency is consistently achieved at a non-trivial $S$ across all benchmarks.

\noindent\textbf{Block size sensitivity to CNN layers.} We also use a series of CNN layers with various input/output channel and kernel combinations to verify our optimal block size selection. Table~\ref{tbl:optimalS} shows our evaluation results with metrics including number of rotation (\#$\textit{Rot}$), amortized memory usage (per image), amortized latency (per image). We select 3 different block sizes for each convolution layer architecture benchmark - Conv($N_{in},N_{out},K$) with corresponding feature map input. To get rid of the impact of ciphertext slot size, we fix the encryption parameter $N=2^{15}$ here. According to our aforementioned analysis (Section~\ref{sec:optimal}), the minimum rotation is achieved when $\frac{K^2}{S^2} =\frac{\alpha \cdot S^2\cdot N_{out}}{N_{in}}$. At these optimal $S^*$, the $\# Rot$ is minimized with the lowest memory usage and latency. For example, Conv(128,32,9) with input feature map (8,128,32,32) achieves the minimum at $S=3.57$ theoretically, thus, in real systems we observe the minimum rotation when block size if set to 4. This observation is aligned with our expectation for the optimal block size $S^*$ and justifies Eq.~\ref{eq:optimal-S}.

\begin{table}[b]
\centering
\caption{Block size v.s. various convolution settings.}
%\vspace{-10pt}
\label{tbl:optimalS}
\resizebox{0.9\columnwidth}{!}{%
\begin{tabular}{c|c|c|c|c}
\hline
\begin{tabular}[c]{@{}c@{}}Layer ($N_{in},N_{out},K$) \\ Input size ($BS,N_{in},H,W$)\end{tabular} 
& \begin{tabular}[c]{@{}c@{}}Block Size\\ $S$ \end{tabular}  
& $\# Rot$ 
& \begin{tabular}[c]{@{}c@{}}Amortized\\ Memory (GB)\end{tabular} 
& \begin{tabular}[c]{@{}c@{}}Amortized\\ Latency (s)\end{tabular} \\ 
\hline

\multirow{3}{*}{\makecell{Conv(64,64,3)\\(4, 64, 32, 32)}}    
& \textbf{1} & 176  & 0.38 & \textbf{0.73} \\ 
& 2          & 288  & 0.63 & 0.78 \\ 
& 4          & 1008 & 2.19 & 1.2 \\ 
\hline

\multirow{3}{*}{\makecell{Conv(64,64,3)\\(8, 64, 32, 32)}}    
& \textbf{1} & 288 & 0.31 & \textbf{0.62} \\ 
& 2          & 320 & 0.34 & 0.72 \\ 
& 4          & 992 & 1.08 & 1.12 \\ 
\hline

\multirow{3}{*}{\makecell{Conv(128,32,9)\\(8, 128, 32, 32)}}    
& 2          & 1648 & 1.79 & 1.89 \\ 
& \textbf{4} & 1008 & 1.09 & \textbf{1.1} \\ 
& 8          & 2224 & 2.42 & 1.34 \\ 
\hline

\multirow{3}{*}{\makecell{Conv(128,32,3)\\(8, 128, 32, 32)}}    
& 1          & 528 & 0.57 & 0.82 \\ 
& \textbf{2} & 304 & 0.33 & \textbf{0.61} \\ 
& 4          & 496 & 0.54 & 0.78 \\ 
\hline

\multirow{3}{*}{\makecell{Conv(32,128,3)\\(8, 32, 32, 32)}}    
& \textbf{1} & 192  & 0.21 & \textbf{0.76} \\ 
& 2          & 496  & 0.54 & 0.86 \\ 
& 4          & 1984 & 2.16 & 0.98 \\ 
\hline

\end{tabular}%
}
% \vskip -0.10in
\end{table}

% \begin{table}[t]
% \caption{Effect of block size $S$ on rotations, memory, and latency.}
% \label{tbl:optimalS}
% \centering
% \resizebox{0.88\columnwidth}{!}{%
% \begin{tabular}{c|c|c|c|c}
% \hline
% Layer / Input & $S$ & \#Rot & Mem (GB) & Lat. (s) \\
% \hline

% \multirow{3}{*}{Conv(64,64,3) / (4,64,32,32)}
% & \textbf{1} & 176  & 0.38 & \textbf{0.73} \\ %\cline{2-5}
% & 2          & 288  & 0.63 & 0.78 \\ %\cline{2-5}
% & 4          & 1008 & 2.19 & 1.20 \\
% \hline

% \multirow{3}{*}{Conv(64,64,3) / (8,64,32,32)}
% & \textbf{1} & 288 & 0.31 & \textbf{0.62} \\ \cline{2-5}
% & 2          & 320 & 0.34 & 0.72 \\ \cline{2-5}
% & 4          & 992 & 1.08 & 1.12 \\
% \hline

% \multirow{3}{*}{Conv(128,32,9) / (8,128,32,32)}
% & 2          & 1648 & 1.79 & 1.89 \\ \cline{2-5}
% & \textbf{4} & 1008 & 1.09 & \textbf{1.10} \\ \cline{2-5}
% & 8          & 2224 & 2.42 & 1.34 \\
% \hline

% \multirow{3}{*}{Conv(128,32,3) / (8,128,32,32)}
% & 1          & 528 & 0.57 & 0.82 \\ \cline{2-5}
% & \textbf{2} & 304 & 0.33 & \textbf{0.61} \\ \cline{2-5}
% & 4          & 496 & 0.54 & 0.78 \\
% \hline

% \multirow{3}{*}{Conv(32,128,3) / (8,32,32,32)}
% & \textbf{1} & 192  & 0.21 & \textbf{0.76} \\ \cline{2-5}
% & 2          & 496  & 0.54 & 0.86 \\ \cline{2-5}
% & 4          & 1984 & 2.16 & 0.98 \\
% \hline

% \end{tabular}}
% \vskip -0.08in
% \end{table}

\noindent\textbf{Efficiency of Architecture-aware Ciphertext Compression (AAC).} We evaluate the effectiveness of our proposed \multi. For this performance evaluation, we adopt two representative convolutional structures as benchmarks: \texttt{residual-shortcuts} from ResNet18 (also used in MobileNet) and the \texttt{fire-modules} from SqueezeNet. These architectures cover a range of design patterns and input dimensions, enabling a comprehensive evaluation of our compression strategy across various convolutional backbones. Table~\ref{tbl:channel_compress_duplication} shows the execution latencies and slot utilization of $4$ fire modules in SqueezeNet and $4$ residual-shorts in ResNet18 on a GPU.
Within each fire module $(N_{\text{in}}, N_{\text{DS}}, N_{\text{out}})$, the input channel $N_{\text{in}}$ is first reduced to a fixed intermediate width $N_{\text{DS}} = 32$, then expanded to half of the output channels via two parallel convolutions, and finally concatenated to form the final output. AAC significantly improves performance by $2.08\times$ by maintaining intermediate cts in a format that maximizes slot utilization. In contrast, non-optimized cts without AAC suffer from increased slot wastage due to dimensional reductions at intermediate stages, leading to degraded slot utilization and higher per-layer latency. The observed speedup ranges from $1.478\times$ to $4.68\times$ across different layers. 

For each residual shortcut module $(N_{\text{in}}, N_{\text{DS}}, N_{\text{out}})$, the input $N_{\text{in}}$ is first projected to $N_{\text{DS}}$ using a $1\times1$ convolution, followed by a $3\times3$ convolution for transformation, and then expanded to $N_{\text{out}}$ by another $1\times1$ convolution. ciphertext-compression yields a total performance gain of $1.636\times$ by maintaining efficient ciphertext packing throughout the intermediate stages. Without ciphertext-compression, the ciphertexts become increasingly underutilized as channel dimensions shrink, resulting in worse slot utilization and increased latency. In this case, the layer-wise speedup ranges from $1.016\times$ to $1.674\times$.
Although the improvement in slot utilization is less pronounced compared to the residual shortcut module, the fire module contains a consecutive $3\times3$ convolutional layer, which amplifies the redundancy caused by wasted slots. As a result, AAC achieves a higher latency improvement in this case.

\noindent\textbf{Repack Overhead Analysis.} We evaluate the latency overhead of additional \textit{rot-mask-add} operations introduced by AAC and repacking to optimize rotation complexity, shown in Figure~{\ref{fig:repack_overhead}}. Specifically, under batch size 1, overhead ranges from 0.42\% (LeNet) to 3.3\% (VGG5), with a similar trend at each benchmark's respective maximum batch size (0.52\%--3.7\%). VGG5 incurs relatively higher overhead due to the more frequent feature map reduction (4/5 layers). While architectures like MobileNet (12/55) and SqueezeNet (2/10) require extensive AAC due to bottleneck modules, the absolute latency contribution of these operations remains marginal. This demonstrates that even for architectures with highly non-uniform layer structures, the overhead of our repacking strategy remains negligible compared to the substantial latency reduction achieved through rotation optimization.
%This demonstrates that even in architectures with highly non-uniform layer structures, the cost of our repacking strategy is strictly subordinate to the massive latency reduction realized through rotation optimization. 
Overall, repacking consistently incurs only minimal overhead across all benchmarks.
%}\bbnote{Q6-C, Q4-B}
%We evaluate the latency overhead of all additional \textit{rot-mask-add} operations, introduced by AAC to enhance slot utilization and repack to enable optimal block-size adjustments across layers, thereby minimizing overall rotation complexity. From the results calibrated on all benchmarks shown in Figure~\ref{fig:repack_overhead}, these operations account for only 0.42\%, 3.3\% 1.2\%, 0.8\%, and 1.8\% of total latency for LeNet, VGG5, SqueezeNet, ResNet18, and MobileNet, respectively, under batch size 1. Even at the maximum batch size, the overhead remains minimal, at 0.52\%, 3.7\%, 1.6\%, 1.1\%, and 2.1\%, demonstrating that repack and AAC performance gains come with negligible additional cost. Notably, VGG5 exhibits the highest relative overhead among our benchmarks, which correlates directly with its architectural characteristics. Specifically, VGG5 undergoes frequent layer shape changes (4 out of 5 layers), significantly more than LeNet (1 out of 5), SqueezeNet (2 out of 10), ResNet18 (4 out of 17), and MobileNet (9 out of 55). (SqueezeNet, MobileNet require more AAC due to more bottleneck modules) While this increased frequency necessitates more frequent repacking, the absolute latency contribution remains marginal. This confirms that even in architectures with non-uniform layer structures, the overhead of our repacking strategy is strictly subordinate to the massive latency reduction achieved through rotation optimization.

\begin{figure}
%\begin{tcolorbox}[colback=blue!10, colframe=blue!30, sharp corners]
    \centering
    \includegraphics[width=1\columnwidth]{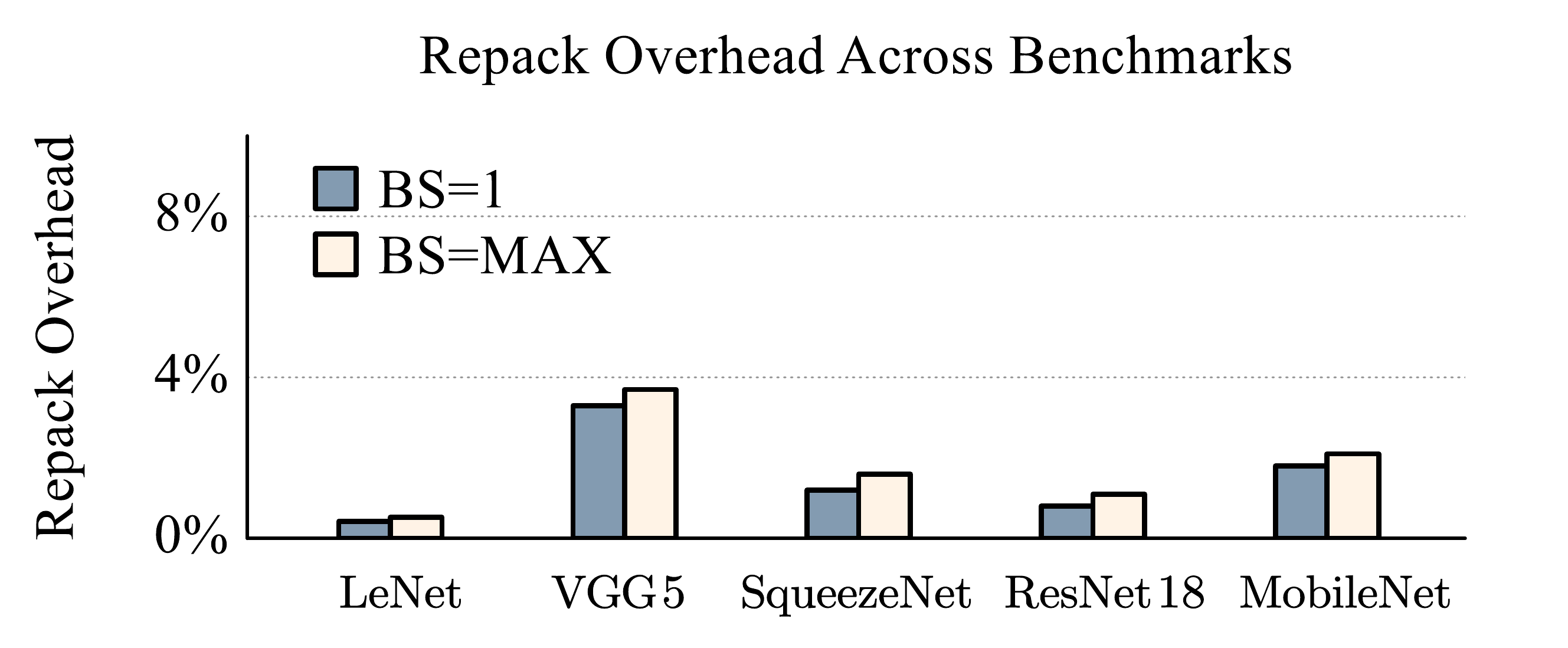}
    %\vspace{-15pt}
   %\caption{End-to-end latency and memory comparison of \ourmodel v.s. prior SOTAs across varying input sizes, batch sizes, and model scales on GPU.}
   \caption{Repack overhead (normalized to total evaluation latency) across various benchmarks. $BS=MAX$ for LeNet, VGG5, SqueezeNet, ResNet18 and MobileNet is 512, 512, 256, 16 and 16, respectively, under 48 GB memory capacity.}
    \label{fig:repack_overhead}
    %\vspace{-5pt}
%\end{tcolorbox}
%\vspace{-15pt}
\end{figure}

 %Overall, their computational overhead is negligible. 

%For SqueezeNet, ResNet18, and MobileNet, we apply AAC to enhance slot utilization, which introduces additional \textit{rot–mask–add} operations or
%However, \textbf{their computational overhead is negligible}. Without generality, for instance, under batch size~1, these operations contribute only $1.2\%$, $0.8\%$, and $1.8\%$ of the total latency for SqueezeNet, ResNet18, and MobileNet, respectively. Even at the maximum batch size, the impact remains similarly low ($1.6\%$, $1.1\%$, and $2.1\%$).

\begin{figure}[b!]
    \centering
    \includegraphics[width=1\linewidth]{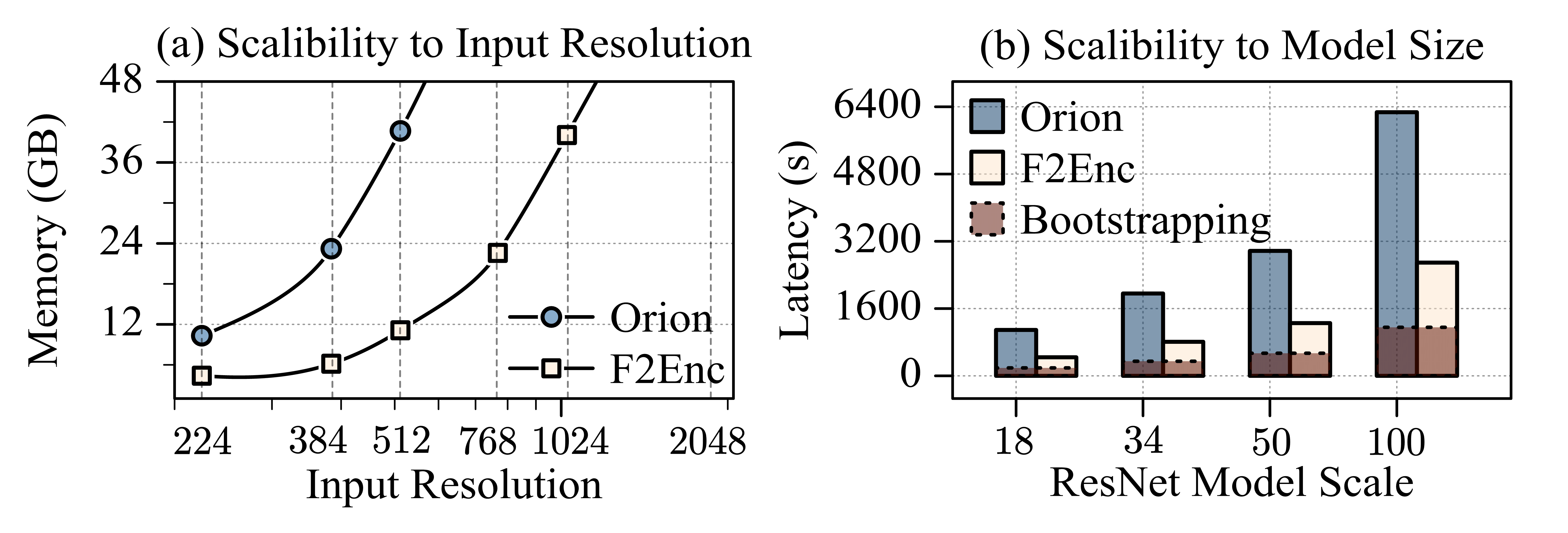}
    %\vskip -0.10in
    \caption{Scalability comparison against \emph{Orion} in terms of higher-resolution input and larger CNN model.}
    \label{fig:scalability-to-model-size-and-input-resolution}
    %\vspace{-5pt}
\end{figure}

\vspace{1mm}
\noindent\textbf{Scalability Analysis.}
As Fig.~\ref{fig:scalability-to-model-size-and-input-resolution} shows, we evaluate the scalability of our method with respect to input resolution and CNN model depth.
When increasing the input size from $224\times224$ to $2048\times2048$, the memory required by \texttt{Orion} grows steeply, whereas \ourmodel exhibits much slower growth thanks to \frag with less amortized rotation cost. Similarly, when scaling from ResNet-18 to ResNet-110, \ourmodel consistently reduces latency with similar performance gap regardless of model depth. Overall, \ourmodel \emph{scales more gracefully to high-resolution inputs and larger CNN models.}

\begin{table}[t]
\scriptsize
\caption{Performance of Arch-aware Ct Compression}
\centering
\label{tbl:channel_compress_duplication}
\resizebox{\columnwidth}{!}{%
\begin{tabular}{l|cc|cc|c}
\hline
\multirow{2}{*}{\begin{tabular}[c]{@{}l@{}}Layer\\ ($N_{in}, N_{DS}, N_{out}$)\end{tabular}} &
  \multicolumn{2}{c|}{without AAC} &
  \multicolumn{2}{c|}{with AAC} &
  \multirow{2}{*}{\begin{tabular}[c]{@{}c@{}}Latency\\ Speedup ($\times$)\end{tabular}} \\ \cline{2-5}
 &
  latency(s) &
  \begin{tabular}[c]{@{}c@{}}slot utilization\end{tabular} &
  latency(s) &
  \begin{tabular}[c]{@{}c@{}}slot utilization\end{tabular} &
   \\ \hline
\begin{tabular}[c]{@{}l@{}}Fire-module\\ (64,32,128)\end{tabular}       & 8.84              & 0.5              & 6.00                & 1              & 1.47          \\ \hline
\begin{tabular}[c]{@{}l@{}}Fire-module\\ (128,32,128)\end{tabular}      & 12.75             & 0.25             & 6.65                & 1              & 1.92          \\ \hline
\begin{tabular}[c]{@{}l@{}}Fire-module\\ (128,32,256)\end{tabular}      & 8.69              & 0.25             & 4.64                & 1              & 1.95          \\ \hline
\begin{tabular}[c]{@{}l@{}}Fire-module\\ (256,32,256)\end{tabular}      & 10.40             & 0.063           & 2.22                & 0.5            & 4.68          \\ \hline
\textbf{Total}                                                          & \multicolumn{2}{c|}{\textbf{40.67}}  & \multicolumn{2}{c|}{\textbf{19.51}}  & \textbf{2.09} \\ \hline
\begin{tabular}[c]{@{}l@{}}Residual-shortcut\\ (64,8,128)\end{tabular}  & 2.35             & 1                & 2.38               &  1              & 1.016          \\ \hline
\begin{tabular}[c]{@{}l@{}}Residual-shortcut\\ (128,8,256)\end{tabular} & 11.87             & 0.25                 & 6.92               &  1              & 1.72          \\ \hline
\begin{tabular}[c]{@{}l@{}}Residual-shortcut\\ (256,8,512)\end{tabular} & 4.12             & 0.06                 & 2.35               &  1              & 1.75          \\ \hline
\begin{tabular}[c]{@{}l@{}}Residual-shortcut\\ (512,8,512)\end{tabular} & 32.17            & 0.02                 & 19.21              &  1             & 1.67          \\ \hline
\textbf{Total}                                                          & \multicolumn{2}{c|}{\textbf{50.50}} & \multicolumn{2}{c|}{\textbf{30.87}} & \textbf{1.64} \\ \hline
\end{tabular}%
}
% \vspace{-10pt}
\end{table}

\section{Discussions}
% \vspace{-2pt}
%\subsection{Applicability of \ourmodel  for Transformers}
\subsection{Extensibility of \ourmodel to Transformers}
The core principle of our \ourmodel extends naturally beyond CNNs to Transformer-based models. In our CNN design, we reduce intra-ciphertext data dependence by splitting feature maps across multiple ciphertexts, thereby decreasing the amount of costly rotation needed within each ciphertext. We apply the same idea to Transformers \emph{by splitting each token embedding across ciphertexts}. Specifically, given $n$ input tokens with embedding dimension $E$, we partition each embedding into sub-blocks of size $S$, producing $E/S$ fragments. The input tensor is thus reorganized from $(n \times E)$ into $(n \times \frac{E}{S} \times S)$, and fragments from different tokens are interleaved across ciphertexts. This fragmented layout reduces embedding-wise intra-ciphertext dependence in the same way that our CNN layout reduces feature-map-wise dependence. As a result, the number of rotations required for embedding aggregation is reduced from $E-1$ in full-embedding packing to $\frac{E}{S}-1$ under fragmented encoding. Importantly, this extension preserves the same trade-off as in the CNN case. The block size $S$ controls the trade-off between inner-rotation and outer-rotation overhead: a smaller $S$ retains more computation within each ciphertext and requires more inner rotations, while a larger $S$ reduces inner-rotation cost at the expense of increased outer-rotation overhead across ciphertexts. This trade-off also appears in Transformer components such as Feed-Forward Network (FFN) and Multi-Head Attention (MHA), where embedding-wise and token-wise rotations are similarly controlled by $S$.

While fragmentation optimizes FFN and projection layers, it introduces a unique trade-off during the MHA stage. Because $S$ times more tokens are packed into each ciphertext to maintain high density, token-wise aggregation (required for $QK^\intercal$ and $AV$ computations) incurs a proportional increase in rotations. The total rotation complexity across both FFN and MHA layers can be formulated as $\mathcal{O}\left( \frac{\gamma E}{S} + S \cdot 2n \right)$, where $\gamma$ is the magnification factor of the FFN layer relative to the embedding size. This reveals a critical optimization space for selecting the block size $S$, analogous to the design-space exploration required for CNNs. 
%Realizing an end-to-end Transformer evaluation requires significant engineering effort beyond linear layer optimization. 

%\mynew{
We implement and evaluate the inference latency corresponding to one encoder block (including MHA and FFN) in the BERT-base model ($E=768, n=128, \gamma \approx 6$). Specifically, to perform the $QK^\intercal V$ computation, we modify the general HE-matrix transposition procedure with specialized token-block rotations to align data between the MHA and FFN stage transition. We additionally implement an adapted Orion packing strategy as the baseline for comparison. 
%Our complexity equation derives an optimal block size of $S=4$. 
We validated this on an NVIDIA A6000 GPU using encryption parameters ($N=2^{16}, \log Q=1768$). As shown in Table~{\ref{tbl:transformer_compare}}, the total latency reaches its minimum at $S=4$, which demonstrates 6.74$\times$ speedup over Orion, confirming that a balanced block size effectively minimizes the combined overhead of embedding-wise and token-wise rotations. Note that evaluating Transformer models end-to-end requires integrating non-linear components such as Softmax and LayerNorm. However, preserving model accuracy when supporting these components under in HE inference remains challenging and is an active research problem{~\cite{zhang2024secure}}. Such an effort is beyond the scope of this work.
\vspace{-5pt}
\subsection{Implications for Hardware Architecture Design}
\label{sec:indication_hw}

%\mynew{
By minimizing homomorphic rotations, our design provides critical insights for FHE hardware specialization. In Ring-LWE-based schemes like CKKS, rotations are the dominant consumers of \textit{NTT/iNTT units}, as the underlying key-switching (or automorphism) procedure requires multiple forward and inverse transforms to transition between decomposition bases. In contrast, other operations, such as ciphertext-plaintext multiplications, primarily involve element-wise polynomial multiplications which are less computationally intensive.

%\mynew{
The drastic reduction in rotation counts through our \ourmodel directly translates to a \emph{lower demand} on the hardware's NTT/iNTT execution units. This shifts the architectural requirements in two significant ways: {\ding{202}} Area and Power Efficiency: Rather than over-provisioning massive NTT clusters to handle traditional rotation-heavy workloads, hardware designers can reduce the number of physical NTT/iNTT engines. This allows for a smaller silicon footprint and significantly lower power consumption without compromising end-to-end throughput; {\ding{203}} Alleviating Memory Pressure: Each (i)NTT invocation necessitates extensive memory traffic to fetch twiddle factors and large ciphertext coefficients. By pruning these invocations, our approach alleviates the memory bandwidth bottleneck that frequently stalls FHE accelerators.

% \begin{itemize}
%     \item \mynew{\textbf{Area and Power Efficiency:} Rather than over-provisioning massive NTT clusters to handle traditional rotation-heavy workloads, hardware designers can reduce the number of physical NTT/iNTT engines. This allows for a smaller silicon footprint and significantly lower power consumption without compromising end-to-end throughput.}
    
%     \item \mynew{\textbf{Alleviating Memory Pressure:} Each (i)NTT invocation necessitates extensive memory traffic to fetch twiddle factors and large ciphertext coefficients. By pruning these invocations, our approach alleviates the memory bandwidth bottleneck that frequently stalls FHE accelerators.}
% \end{itemize}

%\mynew{
The observations from \ourmodel motivates a shift toward hardware accelerators that are provisioned for the \textit{optimal operating region} (e.g., adaptive block sizing) rather than worst-case behavior. The massive gap in kernel calls between traditional and optimized configurations underscores that future FHE hardware should be tailored to exploit such algorithmic optimizations to achieve practical deployment.
%}\bbnote{Q8-C}

% \usepackage[normalem]{ulem}

% \begin{table}[htbp]
% \begin{tcolorbox}[
%     colback=blue!10, 
%     colframe=none, 
%     sharp corners, 
%     boxrule=0pt,
%     left=5pt, right=5pt, top=5pt, bottom=5pt]
% \centering
% \label{tbl:transformer_compare}
% \resizebox{0.9\columnwidth}{!}{
% \begin{tabular}{l|l|l}
% \hline
% Modules                             & Orion    & \ourmodel \\ \hline
% $X\times W_{Q,K,V}$                 & 87.85    & 43.35                   \\ \hline
% $A=Q\times K^{\intercal}$           & 170.0898 & 124.9648                \\ \hline
% $O=A\times V$                       & 32.3868  & 26.347                  \\ \hline
% FFN                                 & 232.715  & 115.8247                \\ \hline
% MatMul latency/Encoder Block        & 523.0416 & 310.4865                \\ \hline
% \end{tabular}
% }
% \end{tcolorbox}
% \end{table}

\begin{table}[tbp]
\centering
\caption{Latency Comparison of Transformer Encoder Block.} 
\label{tbl:transformer_compare}
%\vspace{-10pt}
% \begin{tcolorbox}[
%     enhanced, 
%     colback=blue!10, 
%     frame hidden,
%     sharp corners, 
%     left=5pt, right=5pt, top=2pt, bottom=2pt,
%     hbox,
%     center,
%     nobeforeafter
% ]
% 将列格式从 l|l|l 改为 c|c|c 实现居中
\begin{tabular}{c|c|c|c}
\hline
Modules                             & Orion    & \ourmodel & Speedup  \\ \hline
$X\times W_{Q,K,V}$                 & 87.85    & 43.35   & 2.03$\times$\\ \hline
$A=Q\times K^{\intercal}$           & 1738.17  & 124.96  & 13.91$\times$\\ \hline
$O=A\times V$                       & 32.39    & 26.34   & 1.23$\times$\\ \hline
FFN                                 & 232.71   & 115.82  & 2.01$\times$\\ \hline
MatMul latency/Encoder Block        & 2091.12  & 310.48  & 6.74$\times$\\ \hline
\end{tabular}
%\end{tcolorbox}
%\vspace{-15pt}
\end{table}

\vspace{-3pt}
\section{Related Work}
\vspace{-2pt}
In addition to works discussed in \sectionautorefname~\ref{sec:existing-ciphertext-packing-optim}, orthogonal approaches have explored HE-based inference from different angles. 
At the implementation level, Cheetah~\cite{huang2022cheetah} introduces \textbf{coefficient encoding} using BFV to eliminate rotations via polynomial dot-products; later works such as Hyena~\cite{singh2024hyena}, Bumblebee~\cite{lu2023bumblebee}, and Iron~\cite{hao2022iron} further improve throughput via denser encoding. However, this technique is BFV-specific and requires intermediate data re-encryption, which violates our threat model. NeuJeans~\cite{ju2024neujeans} maps convolution to coefficient-wise dot-sums in CKKS, but its dependence on bootstrapping incurs unnecessary overhead.

At \textbf{compiler level}, HEMET~\cite{lou2021hemet} uses NAS to replace convolutions with HE-friendly operations, targeting mobile models~\cite{iandola2016squeezenet, szegedy2017inception}. Though effective in narrow settings, it lacks generality for broader HE CNNs. In the meanwhile, \textbf{level management} techniques are proposed at compiler level. EVA~\cite{dathathri2020eva} and TenSEAL~\cite{benaissa2021tenseal} help manage noise growth and parameter tuning. BitPacker~\cite{samardzic2024bitpacker} and Reserve~\cite{lee2024performance} optimize level consumption by compressing intermediate results and dynamically reserving levels. Hecate~\cite{lee2022hecate} uses runtime profiling for automatic level budgeting. These works reduce bootstrapping overhead but do not address data packing efficiency. In system-level, various works improved FHE performance through low-level primitive 
acceleration, e.g. NTT/INTT, bootstrapping, and keyswitching, including GPU-based solutions~\cite{fan_tensorfhe_2023, jiao_neo_2025, huang_effact_2025, fan_warpdrive_2025, kim_cheddar_2024, wang_he-booster_2023,gong2026aegis}, and ASIC/FPGA accelerators~\cite{kim_bts_2022, fan_fastfhe_2025, zhou_ufc_2024, kim_anaheim_2025, jayashankar_cinnamon_2025,yang_hydra_2025,samardzic_f1_2021,yang_ola_2025}. 

%\mynew{
Beyond neural network-specific optimizations, recent research in Private Information Retrieval (PIR), such as Pantheon\mbox{~\cite{pantheon2023pantheon}} and Coeus\mbox{~\cite{ahmad2021coeus}}, focuses on a two-step workload consisting of secure matrix-vector scoring followed by parallel oblivious retrieval. Unlike the deep, sequential layers of NNs where circuit depth grows cumulatively, PIR is characterized by a ``wide and shallow" structure where depth consumption is dominated by the initial equality-check protocol. Consequently, PIR represents a bootstrapping-intensive but linear-computation-minor scenario, whereas our work addresses the rotation and memory bottlenecks inherent in deep, sequential models. These approaches are orthogonal; while PIR-specific optimizations like Pantheon's multi-threading focus on parallel execution in CPU, our packing reduces the fundamental rotation complexity per ciphertext, providing a complementary layer of algorithmic efficiency for diverse FHE tasks.
%}\bbnote{Q5-C}

%\vspace{-3pt}
\section{Conclusion}
%\vspace{-2pt}
%\textcolor{red}{TODO: REVISE}
\ourmodel demonstrates that application-level data layout is a first-order architectural optimization for encrypted inference. By jointly optimizing slot utilization, fragment structure, and ciphertext density, \ourmodel dramatically reduces homomorphic operation counts across all convolutional layers. These reductions decrease memory pressure, and improve hardware utilization, yielding up to \textbf{$228.83\times$} latency speedup over \texttt{Orion} on GPU and up to \textbf{$226.06\times$} on CPU for LeNet (MNIST). For MobileNet (ImageNet), \ourmodel achieves up to \textbf{$4.55\times$} latency speedup over \texttt{Orion} on GPU and up to \textbf{$9.43\times$} on CPU, with full memory capacity utilized. Unlike accelerator proposals that target single primitives such as NTT or bootstrapping, \ourmodel reshapes the workload itself, producing benefits that are orthogonal and often significantly larger. \ourmodel is hardware-agnostic yet architecturally transformative, reducing HE workload on hardware by exploiting application-level structure.
%rebalances computation away from memory-bound rotations toward compute-bound NTTs. 
These results highlight a new direction for FHE system design: algorithm-architecture co-design driven by encrypted tensor layout. We believe \ourmodel provides the foundation for next-generation HE accelerators and privacy-preserving ML systems.

\section*{Acknowledgments}
We thank the anonymous reviewers and our shepherd for their constructive feedback and suggestions on this work. This work is partially supported by the National Science Foundation (NSF) under Grants No.~CNS-2348733, No.~CNS-2349538, and No.~CNS-2340777.
%%%%%%% -- PAPER CONTENT ENDS -- %%%%%%%%

%%%%%%%%% -- BIB STYLE AND FILE -- %%%%%%%%
\bibliographystyle{IEEEtranS}
\bibliography{refs}
%%%%%%%%%%%%%%%%%%%%%%%%%%%%%%%%%%%%

\end{document}